\documentclass[onecolumn,draftclsnofoot]{IEEEtran}

\usepackage{latexsym}
\usepackage{amssymb}
\usepackage{bm}
\usepackage{stmaryrd}
\usepackage[dvips]{graphics}
\usepackage{graphicx}
\usepackage{psfrag}
\usepackage{amsfonts,amsmath,amssymb,amsthm}
\usepackage{algorithm}
\usepackage{algorithmic}
\usepackage{dsfont,color,subfigure,setspace}
\usepackage{cite,pdfsync}
\usepackage{float}

\newcommand{\X}{X}
\newcommand{\x}{x}

\newcommand{\Z}{Z}

\newcommand{\T}{T}

\newcommand{\alphX}{\mathcal{\X}}
\newcommand{\indx}{i}

\newcommand{\Xbbf}{{\mathbb{\mathbf{X}}}}

\newcommand{\encod}{\mathcal{E}}
\newcommand{\decod}{\mathcal{D}}

\usepackage{xspace}
\usepackage{bbm}
\newcommand{\ex}{{\rm e}}
\newcommand{\Ac}{\mathcal{A}}
\newcommand{\Bc}{\mathcal{B}}
\newcommand{\Cc}{\mathcal{C}}

\newcommand{\Ec}{\mathcal{E}}
\newcommand{\Fc}{\mathcal{F}}

\newcommand{\Nc}{\mathcal{N}}

\newcommand{\Pc}{\mathcal{P}}
\newcommand{\Qc}{\mathcal{Q}}

\newcommand{\Xc}{\mathcal{X}}
\newcommand{\Yc}{\mathcal{Y}}

\newcommand{\Xv}{{\bf X}}

\newcommand{\xv}{{\bf x}}
\newcommand{\yv}{{\bf y}}
\newcommand{\zv}{{\bf z}}

\newcommand{\Sh}{{\hat{S}}}

\newcommand{\Xh}{{\hat{X}}}

\newcommand{\xh}{{\hat{x}}}

\newcommand{\Xt}{{\tilde{X}}}

\newcommand{\xt}{{\tilde{x}}}

\def\a{\alpha}

\def\d{\delta}
\def\e{\epsilon}

\DeclareMathOperator\E{E}
\let\P\relax
\DeclareMathOperator\P{P}
\newcommand{\Bern}{\mathrm{Bern}}

\newcommand\ie{i.e.,\xspace}
\def\textiid{i.i.d.\@\xspace}
\newcommand\iid{\ifmmode\text{ i.i.d. } \else \textiid \fi}

\newcommand{\ind}{\mathbbmss{1}}

\DeclareMathOperator*{\argmin}{arg\,min}

\newtheorem{theorem}{Theorem}
\newtheorem{lemma}{Lemma}

\newtheorem{corollary}{Corollary}
\newtheorem{definition}{Defintion}

\theoremstyle{definition}
\newtheorem{remark}{Remark}

\begin{document}
\title{Compression-Based  Compressed Sensing}

\author{
    \IEEEauthorblockN{Farideh Ebrahim Rezagah, Shirin Jalali, Elza Erkip, H. Vincent Poor}}

\maketitle

\begin{abstract}

%Efficient
Modern compression algorithms exploit complex structures that are present in signals to describe them very efficiently. On the other hand, the field of compressed sensing is built upon the observation that ``structured'' signals can be recovered from their under-determined set of linear projections.   Currently, there is a large gap between the complexity of the structures studied in the area of compressed sensing and those employed by the state-of-the-art compression codes. Recent results in the literature on deterministic signals aim at bridging this gap    through devising compressed sensing decoders that employ compression codes.
This paper focuses on structured stochastic processes and studies  the application of rate-distortion codes to  compressed sensing of such signals.  The performance of the  formerly-proposed compressible signal pursuit (CSP) algorithm is studied in this stochastic setting. It is proved that in the very low distortion regime, as the blocklength grows to infinity, the CSP algorithm reliably and robustly  recovers $n$ instances of a stationary  process  from random linear projections as long as their count is slightly more than $n$ times the rate-distortion dimension (RDD) of the source. It is also shown that under some regularity conditions, the RDD of a stationary process is equal to its information dimension (ID).  This connection establishes the optimality of the CSP algorithm  at least for memoryless stationary sources, for which  the fundamental limits are known. Finally, it is shown that the CSP algorithm combined by a family of universal variable-length fixed-distortion compression codes yields a family of universal compressed sensing recovery algorithms. 
%where the fundamental limits of CS are known.

%\blfootnote{This research was supported in part by the U. S. National Science Foundation under Grant CCF-1420575.}\\
%\\

\end{abstract}
\section{Introduction}

Consider the standard  setup of a compressed sensing data acquisition system: a decoder observes a noisy linear projection of  the high-dimensional signal $\xv\in\mathds{R}^n$, \ie $\yv=A\xv+\zv$, where $A\in\mathds{R}^{m\times n}$, $m<n$, is the measurement matrix, and $\zv\in\mathds{R}^m$ denotes the measurement noise. The signal is assumed to be ``structured'', which typically means that it is sparse in some transform domain.
%, otherwise it is impossible for the decoder to recover it from the measurements. %The goal of the decoder is to estimate the signal $\xv$ from the measurements  $\yv$. 
 The decoder is expected to recover the signal $\xv$ using a computationally efficient algorithm with as few number of   measurements, $m$, as possible. Such modern data acquisition problems, which can be described as  solving    under-determined systems   of linear equations,  arise in many different applications, including  magnetic resonance imaging (MRI), high resolution imaging, and radar.

 While sparsity of the desired signal is the main focus in the compressed sensing literature, started by the key works of  Donoho et al. ~\cite{Donoho:06} and Candes et al. \cite{CandesT:06,CandesR:06},  more recent compressed sesning recovery algorithms capture structures beyond sparsity, such as group sparsity, low-rankness, etc.~\cite{RichModelbasedCS, VeMaBl02, ReFaPa10, HeBa11, DoKaMe06,BakinThesis, eldar2010block, YuLi06, ji2009multi, MaAnYaBa11, stojnic2009block, stojnic2009reconstruction, stojnic2010ell, MeVaBu08, ChRePaWi10,  som2012compressive, donoho2013microlocal, CaLiMaWr09,  waters2011sparcs, ChSaPaWi11, duarte2011performance, blumensath2009sampling, mccoy2012sharp, studer2011stable, peyre2011group}. Although the studied structures in the literature and their extensions are present in many signals of interest, and yield  promising results, they are to a great extent confined to  basic models,
compared to more complex underlying structures  known to be present in such signals.  Employing such elaborate structures that are usually present in the signals can potentially lead to much more efficient compressed sensing systems that require significantly smaller numbers of measurements, for achieving  the same reconstruction quality.
%
%Another approach to the CS problem has been taken by Wu et.al. \cite{WuV:10} where the fundamental limits of CS is studied. They show that the minimum number of measurements required for almost lossless reconstruction of i.i.d. processes is equal to the ID of the source, regardless of the measurement and recovery algorithms used. In other words, ID quantifies the effect of any underlying structure that exists in the source in reducing the number of measurements required for almost recovery of the source.

In addition to compressed sensing, the structure of a signal plays  an important role in many  other fundamental problems in information theory, such as data compression, data prediction and denoising. Data  compression  is a well-studied topic in information theory,  initiated  by the Shannon's seminal work \cite{Shannon:48}. Compression algorithms employ the patterns  in a signal to render efficient digital representation of it. After decades of research, the types of structure employed by the state-of-the-art compression algorithms, especially for coding image, audio and video signals, are quite elaborate, and much more complicated than those studied in compressed sensing. %Therefore, designing signal recovery algorithms that take advantage of those complicated structures potentially leads to much more efficient CS-recovery algorithms.

%In this paper we attempt to utilize \emph{compressability} as the underlying structure for the recovery of an undersampled random process with the goal of finding the number of measurements required for reconstruction of the original process satisfying the fidelity criterion.The idea is that if a set of signals can efficiently be compressed, then the signals in that set are  structured.
%%\cite{Donoho:06,CandesT:06,CandesR:06,RichModelbasedCS, ChRePaWi10, VeMaBl02, ReFaPa10, ShCh11, HeBa12, HeBa11, DoKaMe06}
%Given a compact set and a family of compression codes, \cite{JalaliM:15-acha} and \cite{JalaliM:13-isit} introduce the  \emph{compressible signal pursuit} (CSP) algorithm to utilize the underlying structure.

Given the maturity and the efficiency of existing data compression algorithms,  one may wonder whether  data compression codes can be directly  employed to build compressed sensing recovery algorithms. The motivation for such an approach is that in order for  a good compression code to represent some process as efficiently as possible,  theoretically, it should employ all the structure that is present in it. Therefore, building a compressed sensing decoder based on an efficient data compression code, potentially might enable the decoder to exploit all the structure present in the data, and, thereby minimize the number of measurements. Also, another advantage of this approach would be devising a generic process for building compressed sensing recovery algorithms based on compression codes to simplify this task.  Hence, to  find the most efficient compressed sensing data recovery algorithm for a given source of data,  instead of studying and learning  some specific structure in the source model,  the already existing data compression codes can be used to  directly  design efficient compressed sensing decoders. 

The idea of utilizing compression codes in designing compressed sensing recovery algorithms was introduced in  \cite{JalaliM:13-isit} and \cite{JalaliM:15-acha}.  Consider $x^n \in \Qc$, where $\Qc$ represents  a compact subset of $\mathds{R}^n$. A compression code of rate $r$ for the set $\Qc$ is described by encoder and decoder mappings
\[
f_n:\Qc\to \{1,\ldots,2^{r}\}
\]
and
\[
g_n:  \{1,\ldots,2^{r}\}\to \mathds{R}^n,
\]
respectively. The distortion induced by this code is defined as
\[
\d\triangleq\sup_{x^n\in\Qc}\|\x^n-g_n(f_n(x^n))\|_2.
\]
This code defines a codebook $\Cc_n$, which contains all possible reconstruction vectors generated by this code. That is,
\[
\Cc_n\triangleq \{g_n(f_n(x^n)): x^n\in\Qc\}.
\]
Clearly, $|\Cc_n|\leq 2^r$. %Assume that a decoder  observes $y_o^m=Ax_o^n+z^m$, where $x_o^n\in\Qc$ and  $z^m\in\mathds{R}^m$ denote  the signal and  the measurement noise, respectively.  
Suppose a decoder desires to recover the signal $x^n$ from noisy underdetermined linear projections $y^m=Ax^n+z^m$,  by employing the compression code $(f_n,g_n)$, without explicitly studying the  set $\Qc$.  To achieve this goal, \cite{JalaliM:13-isit} and\cite{JalaliM:15-acha} propose the compressible signal pursuit  (CSP) optimization  defined as
\[
\tilde{x}^n\triangleq \argmin_{c^n\in\Cc_n}\|y^m-Ac^n\|_2^2 .
\]
In other words, to recover the original signal from sufficient number of random linear projections,  through an exhaustive search over the codebook of the compression code, the CSP  seeks the reconstruction vector in  $\Cc_n$ that minimizes  the measurement error. It can be shown that the  required number of measurements for successful recovery  depends on the rate-distortion trade-off of the compression code and the desired  accuracy \cite{JalaliM:13-isit,JalaliM:15-acha}.

%\textcolor{blue}
{The results of \cite{JalaliM:13-isit} and \cite{JalaliM:15-acha} on deterministic signals establish the foundations of building compression-based compressed sensing decoders. However, since the studied model only concerns  deterministic signals, the results do not illustrate the fundamental connections between the source structure, which is captured by its distribution, its information theoretic rate-distortion function and the number of measurements required by  compression-based decoders. In this paper we focus on \textit{stationary analog processes}, and study the performance of the CSP algorithm, as a compression-based compressed sensing recovery algorithm. This shift from deterministic signals to stochastic stationary processes enables us to 
\begin{enumerate}
\item characterize the performance of the CSP algorithm in terms of the information theoretic rate-distortion function of the source, and illustrate the connection between the asymptotic number of measurements required by the CSP and the rate-distortion dimension of the source process;
\item  establish new fundamental connections between the rate-distortion dimension of the source, and its information dimension, which serves as its measure of complexity;
\item employ the established connection and prove asymptotic optimality of the CSP algorithm for cases in which the fundamental limits of compressed sensing is known; and
\item employ universal compression codes, and design a compression-based universal compressed sensing recovery algorithm.
\end{enumerate} }

Since the sources of interest in compressed sensing applications are usually  analog,  compression codes employed in building  compression-based decoders has to be  lossy codes. As a result, the reconstruction given by the CSP algorithm is also a lossy reconstruction. In other words, the resulting compression-based recovery algorithm is a lossy compressed sensing algorithm, where there is a trade-off between the number of measurements, the  quality  of the reconstruction, and the rate and the distortion  of the compression code. In this paper we  mainly focus on the this trade-off and  leave the complexity issues for future extensions of this work, where a more algorithmic approach would be necessary to handle or at least approximate the minimization  in CSP with reasonable time-complexity.

In a standard compressed sensing setting, the decoder recovers the signal \emph{losslessly} or almost losslessly from an  underdetermined set of linear equations. While compression-based recovery algorithms enable  us to exploit more complex structures, there is an inherent loss due to the underlying lossy compression codes. By letting the distortion of the compression code become arbitrarily small, we can achieve almost lossless recovery which is of interest in compressed sensing problems. Although arbitrarily small distortion for an analog signal dictates an arbitrarily large compression code rate, the RDD of the code,  which is the quantity that directly relates the compression code's rate-distortion behavior to the number of required measurements, remains bounded.

%---------------------%---------------------%---------------------%---------------------

%\subsection{Contributions}

In this paper we consider a stochastic analog source $\Xbbf=\{X_i\}_{i=-\infty}^{\infty}$ and signal $X^n$ generated by this source. Instead of observing $X^n$ directly, a decoder measures $Y^m=AX^n+Z^m$, $m<n$, and aims at estimating $X^n$ from $Y^m$. Here, similar to the deterministic setup, $A\in\mathds{R}^{m\times n}$ and $Z^m$ denote the measurement matrix and the stochastic noise in the system, respectively.  Assume that the data acquisition decoder
 %, instead of the distribution of the source $\Xbbf$, 
has access to  a ``good'' lossy compression code for the source $\Xbbf$, and employs it to recover the vector $X^n$ via the CSP algorithm. Our first major contribution in this paper is to derive the trade-off between the performance of the compression code, stated  in terms of its rate, distortion and  excess distortion probability, and the performance of the CSP algorithm, summarized by  the required  number of linear measurements and its achieved reconstruction quality. We prove that, asymptotically, for large $n$ and as the distortion of the compression codes goes to zero, the normalized  number of random linear measurements required by the CSP algorithm is equal to the RDD \cite{KawabataD:94} of the source. It is known that for a random variable (or vector), the (upper and lower) RDD is equal to the (upper and lower) information dimension (ID) of the random variable \cite{KawabataD:94}. Our second major contribution is to extend this result to analog stationary processes, and to prove that, under some regularity conditions,  the  RDD of a stationary process is equal to its ID, defined in \cite{JalaliP:14-arxiv}. This combined with the results of \cite{WuV:10} establishes the asymptotic optimalilty of CSP %compression-based compressed sensing approach 
for stationary memoryless sources.

We study  piecewise-constant signals to illustrate our results on the connection between RDD and ID.  Piecewise-constant signals are used widely to model many natural signals in the signal processing, compression, and denoising literature. We derive upper and lower bounds on the rate-distortion functions of such signals, when they are modeled by a first-order Markov process and use these bounds to  evaluate the  RDD of such processes.

%In information theory, universal codes refer  to the algorithms that do not require the knowledge of the source distribution and yet achieve the optimal performance. Universal lossy or lossless compression \cite{LZ77,LZ,Sakrison:70,Ziv:72,NeuhoffG:75,NeuhoffS:78,Ziv:80,GarciaN:82}, universal denoising \cite{kolmogrov_sampler,dude} and universal prediction \cite{MerhavGutmanFeder92,MerhavF:98} are some examples of different universal coding problems that are well-studied in information theory. The problem of universal compressed sensing and the existence of such algorithms  has recently been studied both for deterministic \cite{JalaliM:14} and probabilistic signal models \cite{JalaliP:14-arxiv}. In this paper, given our focus  on building compression-based compressed sensing algorithms, we address a related important question: can one derive a universal compressed sensing recovery algorithm based on a given universal compression code? How well will such scheme perform? Our third major contribution is addressing both of these questions. We prove that a family of universal fixed-distortion compression codes yields a family of universal compressed sensing recovery algorithms. This connection has important implications both in theory and in practice.

Given our focus  on building compression-based compressed sensing algorithms, we also address two related important questions: Can one derive a universal compressed sensing recovery algorithm based on a given universal compression code? How well will such a scheme perform? 
In information theory, universal codes refer to algorithms that do not require knowledge of the source distribution and yet achieve the optimal performance. Universal lossy or lossless compression \cite{LZ77,LZ,Sakrison:70,Ziv:72,NeuhoffG:75,NeuhoffS:78,Ziv:80,GarciaN:82}, universal denoising \cite{kolmogrov_sampler,dude} and universal prediction \cite{MerhavGutmanFeder92,MerhavF:98} are some examples of  universal coding problems that have been well-studied in information theory. The problem of universal compressed sensing and the existence of such algorithms that can recover a signal from its underdetermined set of random linear observations without knowing the source model has recently been studied both for deterministic \cite{JalaliM:14} and probabilistic signal models \cite{BaronD:11,BaronD:12,JalaliP:14-arxiv}. 
Our third major contribution is addressing both of the above questions. We prove that a family of universal fixed-distortion compression codes yields a family of universal compressed sensing recovery algorithms. This connection has important implications both in theory and in practice. 

%and can be quantified as the ratio of the rate of the codebook over the logarithm of distortion it's achieving, which has a close relationship to rate-distortion-dimension (RDD) \cite{KawabataD:94} of the compression code.

%              Talk about the WuVerdu results and Universal results.

The organization of the paper is as follows. Section \ref{sec:CSP-performance} studies  the performance the CSP algorithm when applied to compressed sensing of a stationary process.  Section \ref{sec:RDD} examines the properties of complexity measures for analog stationary processes, and establishes  a connection between the ID  and the RDD of such processes and also provides bounds on the rate-distortion region of the  piecewise constant source modeled by a first-order Markov process  to illustrate this relationship.   Section \ref{sec:AlmostLosslessCSP} provides the performance and optimality of CSP for almost lossless recovery using the established connection between RDD and ID. Universal CSP (UCSP) is introduced in Section \ref{sec:UCSP-performance} as a universal compressed sensing recovery algorithm, and its performance trade-offs are studied.  %The rate-distortion region for piecewise constant processes is studied in Section \ref{sec:numericalresults} and CSP performance trade offs for one such source is illustrated. 
Section \ref{sec:proofs} presents the proofs of some of the results, 
and Section \ref{sec:conc} concludes the paper.

%---------------------%---------------------%---------------------%---------------------

\subsection{Notation}

Calligraphic letters such as $\Xc$ and $\Yc$ denote sets. The size of a set $\Xc$ is denoted by $|\Xc|$. Capital letters like $X$ and $Y$ represent random variables. For a random variable $X$, $\Xc$ denotes its alphabet.  For $x\in\mathds{R}$, $\lceil x \rceil$ ($\lfloor x \rfloor$) represents  the smallest (largest) integer larger (smaller) than $x$. For $b\in\mathds{N}^+$, $[x]_b$ denotes the $b$-bit approximation of $x$, \ie for $x=\lfloor x \rfloor+\sum_{i=1}^{\infty}(x)_i2^{-i}$, $(x)_i\in\{0,1\}$,
\[
[x]_b=\lfloor x \rfloor+\sum_{i=1}^{b}(x)_i2^{-i}.
 \]Also, let $\langle x\rangle_b$ defined as
 \[\langle x\rangle_b={\lfloor bx\rfloor \over b},\]
denote the discretized version of $x$. For $x\in\mathds{R}$, $\delta_x$ denotes the Dirac measure with an atom at $x$. Throughout the paper, $\log$ and $\ln$ refer to the logarithm in base 2 and natural logarithm, respectively. $\{0,1\}^*=\cup_{n=1}^{\infty}\{0,1\}^n$ denotes the set of all binary sequences of finite length. For a binary sequence $b\in\{0,1\}^n$, $|b|$ denotes the length of the sequence.

\section{Compressible signal pursuit}\label{sec:CSP-performance}

%Consider the CS problem in which the goal is to recover a \emph{structured} signal $\x^n$ from its undersampled set of linear measurements $Y^m=A\x^n$ $(m<n)$. For certain types of structure, such as sparsity, it has been shown that $\x^n$ can be recovered efficiently and robustly  from measurements $Y^m$ even when $m<n$ \cite{Donoho:06,CandesT:06,CandesR:06,RichModelbasedCS, VeMaBl02, ReFaPa10, HeBa11}.

This section extends the CSP algorithm proposed in  \cite{JalaliM:13-isit} and \cite{JalaliM:15-acha} to stochastic processes. The intuition behind the CSP algorithm is  that if  a set of signals can  be compressed efficiently using a compression code, then the structure employed by the compression code can indirectly, through the application of the compression code, be used in building efficient compressed sensing recovery algorithms. In other words, the CSP algorithm, through the compression code, extracts all the useful structure present in the data to reduce the number of linear measurements.

%Consider compact set $\Qc\subset\mathds{R}^n$, and compression code described by encoder and decoder mappings  $f_n:\Qc\to \{1,\ldots,2^{r}\}$ and $g_n:  \{1,\ldots,2^{r}\}\to \mathds{R}^n$.  The distortion induced by this code is defined as $\d\triangleq\sup_{x^n\in\Qc}\|\x^n-g_n(f_n(x^n))\|_2.$
%The codebook of this code is denoted by $\Cc_n=\{g_n(f_n(x^n)): x^n\in\Qc\}$. To recover signal $x_o^n\in\Qc$ from its underdetermined linear projection $y_o^m=Ax_o^n$, \cite{JalaliM:15-acha,JalaliM:13-isit} propose the  \emph{compressible signal pursuit} (CSP) algorithm defined as
%\[
%   \tilde{x}^n_o=\argmin_{x^n\in\Cc_n}\|y_o^m-Ax^n\|_2^2 .
%\]

Consider a random vector $X^n$, generated by stationary process $\Xbbf=\{\X_i\}_{\indx=0}^\infty$, where $X_i\in\Xc$.  A compressed sensing decoder observes  a  linear projection of $X^n$, 
\[
Y^m=A X^n,
\]
where $A\in{\rm I\!R}^{m\times n}$ denotes the measurement matrix with $m<n$, and aims at estimating $X^n$. % In this section, we study the performance of the CSP algorithm, when applied to a stationary analog process as a compression-based lossy compressed sensing recovery algorithm.

%Consider a lossy compression code $(n,f_n,g_n)$ operating at rate $R$ on the above process. In other words, $f_n:\Xc^n\to \{1,\ldots,2^{nR}\}$ and $g_n: \{1,\ldots,2^{nR}\}\to \hat{\Xc}^n$. Here, $\hat{\Xc}$ denotes the reconstruction alphabet, which is usually equal to $\Xc$. For source vector $X^n$, $\Xh^n=g_n(f_n(X^n))$ denotes its lossy reconstruction produced by the $(n,f_n,g_n)$ code. Throughout the paper, we assume that the reconstruction quality of decoders is evaluated using the single letter squared error distortion  measure, i.e. $d(X^n, \Xh^n) = \frac{1}{n}\|X^n-\Xh^n\|^2_2$.   Assume that the  $(n,f_n,g_n)$ code achieves  distortion $D$ with excess distortion probability $\e$. That is,
%\[
%\P({1\over n}\|X^n-\Xh^n\|^2_2>D)\leq \e.
%\]
% Let $\Cc_n\triangleq \{g_n(f_n(x^n)): \; x^n\in\Xc^n\}$ denote the codebook of this compression code. Clearly, $|\Cc_n|\leq 2^{nR}$. Given the source output $\X^n$ and the  observation vector $Y^m=A\X^n$, let   $\Xt^n$ denote the solution of the CSP algorithm employing the $(n,f_n,g_n)$ code. In other words,
%\begin{align}
%\Xt^n=\argmin_{x^n\in\Cc_n}\|Y^m-Ax^n\|_2^2. \label{CSPsolution}
%\end{align}

%\begin{remark}

%________________________________________________________________%
%For a stationary process $\mathbf{X}=\{X_i\}_{i=-\infty}^{\infty}$, let  $\Xc$ and $\hat{\Xc}$  denote the source and reconstruction alphabets, respectively.
A fixed-length lossy compression code for the source $\mathbf{X}$ operating at rate $R$ and blocklength $n$ is specified as $(n,f_n,g_n)$, where 
\[
f_n:\alphX^n\mapsto\{1,2,\ldots,2^{nR}\}
\]
and   
\[
g_n:\{1,2,...,2^{nR}\}\mapsto\hat\alphX^n
\] 
denote the encoding and  the decoding functions, respectively and $\hat{\Xc}$  is the reconstruction alphabet. Throughout the paper, we mainly focus on the case where $\cal{X}=\hat{\cal{X}}={\rm I\!R}$ with squared error distortion $d(x,\hat{x})=(x-\hat{x})^2$, where, $d: \cal{X}\times \hat{\cal{X}}\to {\rm I\!R}^+$ denotes a per-letter distortion measure.
Traditionally, the performance of a lossy compression code is measured  in terms of its rate, $R$, and expected average distortion  $D\triangleq \E[d_n(X^n,\Xh^n)],$ where $\Xh^n=g_n(f_n(X^n))$, and 
\[
d_n(x^n,\hat{x}^n)\triangleq {1\over n}\sum_{i=1}^nd(x_i,\hat{x}_i),
\] 
for any $x^n\in\Xc^n$ and $\xh^n\in\hat{\Xc}^n$.
Another possible performance metric for a lossy compression code is its \emph{excess distortion  probability} \cite{Marton:74}, which is a stronger notion than expected distortion. The excess distortion probability of a code is defined as the probability that the average per-letter distortion between the source and reconstruction blocks exceeds some predetermined threshold, \ie $\P(d_n(X^n,\Xh^n) > D ).$ 
%In other words, the  $(n,f_n,g_n)$ code achieves  distortion $D$ with excess distortion probability $\e>0$, if $\P({1\over n}\|X^n-\Xh^n\|^2_2>D)\leq \e$, for large enough $n$. 
The distortion $D$ is said to be achievable at rate $R$ if for any $\e>0$, there is a large enough $n_0$ such that for any $n>n_0$ the $(n,f_n,g_n)$ code  satisfies 
\[
\P(d_n(X^n,\Xh^n) > D )\leq\e,
\] 
i.e. the excess distortion probability $\e$ can be driven to zero as $n\to\infty$.

\begin{remark}

Let $R_{m}(\mathbf{X},D)$ and  $R_{a}(\mathbf{X},D)$ denote the rate-distortion functions of a source $\Xv$ under vanishing excess distortion probability and expected average distortion, respectively. While $R_{m}(\mathbf{X},D)$ and  $R_{a}(\mathbf{X},D)$ are not equal in general, for stationary ergodic processes $R_{m}(\mathbf{X},D)=R_{a}(\mathbf{X},D)$ \cite{SteinbergV:96,IharaM:05,Iriyama:05}. Throughout this paper we focus only on such processes; therefore, we drop the subscript $m$ or $a$, and let $R(\mathbf{X},D)$ denote the rate-distortion function of the source. 

%Furthermore, when it is clear from the context, $R$ and $R(\mathbf{X},D)$ will be used interchangeably to simplify the notation.
%________________________________________________________________%

\end{remark}

Let 
\[
\Cc_n\triangleq \{g_n(f_n(x^n)): \; x^n\in\Xc^n\}
\] 
denote the codebook of this compression code. Clearly, $|\Cc_n|\leq 2^{nR}$. Given the source output $\X^n$ and the  observation vector $Y^m=A\X^n$, let   $\Xt^n$ denote the solution of the CSP algorithm employing the $(n,f_n,g_n)$ code. In other words,
\begin{align}
\Xt^n=\argmin_{x^n\in\Cc_n}\|Y^m-Ax^n\|_2^2. \label{CSPsolution}
\end{align}

The following theorem derives an upper bound on the loss incurred by the CSP in recovering $X^n$. The bound on reconstruction distortion holds with high probability and depends on the parameters of the compression code $n$, $R$, $D$ and $\epsilon$, and the number of measurements $m$. It is important to note that the compression code used by the CSP algorithm  is not required to be an optimal code, and the theorem also holds even if the CSP algorithm is based on an off-the-shelf compression code.

\begin{theorem}\label{NoiselessThm}
Consider $Y^m=A X^n$, a system of random linear observations  with measurement matrix $A\in {\rm I\!R}^{m\times n}$, where $A_{i,j}$ are independently and identically distributed (i.i.d.)~ as $\mathcal{N}(0,1)$. Let $\Cc_n$ be a lossy compression code for $X^n$ operating at rate $R$ that achieves  distortion $D$ with excess distortion probability $\epsilon$. 
Without any loss of generality assume that the source is normalized such that $D<1$. For arbitrary $\a>0$ and $\eta>1$, let $\d={\eta\over \log{1\over D}}+\a,$ and
%and $A\in {\rm I\!R}^{m\times n}$, where
\[
{m\over n}=\frac{2\eta R}{\log{1\over D}},
\] be the normalized number of observations. Let  $\Xt^n$ denote to the solution of the CSP algorithm given in \eqref{CSPsolution}.
Then,
\begin{align*}
\P&\Big(\frac{1}{\sqrt{n}}\lVert\; X^n - \Xt^n\rVert_2 \geq (2+\sqrt{n/m})D^{\frac{1}{2}(1-\frac{1+\d}{\eta})}  \Big)\\
&\;\;\leq \e+2^{-\frac{1}{2}nR\alpha}+ \mathrm{e}^{-{m\over 2}}.
\end{align*}
\end{theorem}

\begin{proof}
The proof %to Theorem \ref{NoiselessThm},
is provided in Section \ref{sec:proofs}.
\end{proof}

Theorem \ref{NoiselessThm} states that using a class of compression codes $\Cc_n$ operating at rate $R$ and distortion $D$, %for which the excess distortion probability $\e$ can be driven to zero as $n$ goes to $\infty$
with  $m=\frac{2\eta Rn}{\log(1/D)}$ random linear measurements ($\eta>1$) of $n$ samples of a stochastic process, the distortion incurred by CSP in recovering $X^n$ can be upper-bounded with probability approaching one as $n$ grows without bound.  In the limit when $D$ approaches zero, the normalized number of measurements required by the CSP algorithm, for almost lossless recovery of the source, depends on the limit of $\frac{ 2R}{\log(1/D)}$. If the compression code used by the CSP operates close to the fundamental rate-distortion tradeoff of the source, this limit approaches the rate-distortion dimension of the source \cite{KawabataD:94}. To better understand the performance of the CSP algorithm, in the following section, we focus on this quantity and explore its connections with other known measures of complexity for stationary processes.

%\begin{remark}
As stated in Theorem \ref{NoiselessThm}, $\eta>1$ is a free parameter that affects the performance of the CSP algorithm.  Choosing a small $\eta$, arbitrarily close to $1$, minimizes the number of random linear measurements, $m$, required by the CSP. On the other hand, since the reconstruction distortion scales as $D^{{\frac{1}{2}(1-\frac{1+\d}{\eta})}}$, where $\d>0$, for optimal scaling of the distortion $(\sqrt{D})$, $\eta$ needs to be large. In other words, the closer ${\frac{1}{2}(1-\frac{1+\d}{\eta})}$ gets to $1$, the better performance we get from CSP in terms of reconstruction distortion. Therefore, as $\eta$ varies, there is a trade-off between the number of measurements on one hand and the scaling of the reconstruction distortion on the other hand. 
%\end{remark}

%%--------------------------%------------------------------------%-----------------------
%\section{Noisy measurements}

Theorem \ref{NoiselessThm} characterizes the performance of the CSP algorithm in recovering a random process, when there is no noise in the measurement process. In reality, there is always some noise in the system.  The following theorem proves the robustness of  the performance of the CSP algorithm to measurement noise. Specifically assume that instead of $Y^m=AX^n$, the decoder observes $Y^m=AX^n+Z^m$, where $Z^m$ denotes some random measurement noise. Further assume that the decoder employs the CSP algorithm as before to recover $X^n$ from measurements $Y^m$. That is, $\Xt^n$ is still given by \eqref{CSPsolution}. The following theorem states that if the noise power is not very large and the compression code's distortion $D$ stays away from zero, then the performance of the CSP algorithm essentially stays the same.

\begin{theorem}\label{Noisy Thm}
Consider $Y^m=A X^n+Z^m$, a noisy system of random linear observations where $Z^m$ is the additive noise and $A\in {\rm I\!R}^{m\times n}$ is the measurement matrix where $A_{i,j}$ are i.i.d.~as $\mathcal{N}(0,1)$. Assume that the average power of the noise can be bounded by $\sigma_m^2$ with probability $1-\e_m$, i.e.
\[
\P\Big({1\over \sqrt{m}}\|Z^m\|_2>\sigma_m\Big)<\e_m.
\]

Let $\Cc_n$ be a lossy compression code for $X^n$ operating at rate $R$ that achieves  distortion $D$ with excess distortion probability $\epsilon$. 
Without any loss of generality assume that the source is normalized such that $D<1$. For arbitrary $\a>0$ and $\eta>1$, let $\d={\eta\over \log{1\over D}}+\a,$ and $m=\frac{2\eta nR}{\log(1/D)}$ be the normalized number of observations, and let $\Xt^n$ be the solution of the CSP algorithm, as given by \eqref{CSPsolution}. Then,
\begin{align*}
P&\bigg(\frac{1}{\sqrt{n}}\lVert\; X^n - \Xt^n\rVert_2
\geq \nonumber\\
&\hspace{0.5cm} \; (2+\sqrt{n/m})D^{\frac{1}{2}(1-\frac{1+\d}{\eta})} +{2\sigma_m \over  \sqrt{D^{1+\d \over \eta}n}}\bigg)\nonumber\\
&\hspace{0.5cm} \leq \e_m+\e+2^{-{1\over 2}nR\alpha} + \mathrm{e}^{-\frac{m}{2}}.
\end{align*}
%with probability exceeding $1-\e-\e_m-2^{-{1\over 2}nR\alpha}- \mathrm{e}^{-\frac{\eta nR}{\log(1/D)}}$.
\end{theorem}
\begin{proof}
The proof is provided in Section \ref{sec:proofs}.
\end{proof}

The effect of the noise on the error is captured  by the term, ${2\sigma_m \over  \sqrt{D^{1+\d \over \eta}n}}$, which disappears as $n\to\infty$. This is due to the fact that by drawing the entries of the measurement matrix based on an i.i.d. $\Nc(0,1)$ distribution, the signal to noise ratio (SNR) of each measurement goes to infinity as $n\to\infty$. If instead the entries of $A$  are drawn as $\Nc(0,{1\over n})$, then the term dependent on noise becomes  ${2\sigma_m /  \sqrt{D^{1+\d \over \eta}}}$, which does not disappear as $n$ grows to infinity.  
%Since ${2\sigma_m /  \sqrt{D^{1+\d \over \eta}}}$ is inversely proportional with $D$, in a noisy system, choosing $D$ very close to zero leads to a big penalty due to the noise. This we believe is an artifact of our proof technique and is not something fundamental. 

%\begin{corollary}
%Consider the setup of Theorem \ref{Noisy Thm}, but assume that  $Z^m$ is drawn from $\Nc(0,\sigma^2 I)$. Consider $\tau'>0$. Then,
%\begin{align*}
%&\frac{1}{\sqrt{n}}\lVert\; X^n - \Xt^n\rVert_2 \nonumber\\
%&\leq \; D^{\frac{1}{2}(1-\frac{1+\d}{\eta})} (2+\sqrt{n/m})+{4\sigma \over  \sqrt{D^{1+\d \over \eta}n} },
%\end{align*}
%with probability exceeding $1-\e-\ex^{-0.8 m}-2^{-{1\over 2}nR\alpha}- \mathrm{e}^{-\frac{\eta nR}{\log(1/D)}}$.
%
%\end{corollary}
%\begin{proof}
%Since $\frac{1}{\sigma}Z^m\sim\Nc(0,I)$, by $\chi^2$-construction lemma we get:%by Lemma \ref{chi2},
%\begin{align*}
%P(\frac{1}{\sigma}\|Z^m\|_2>\sqrt{m(1+\tau\rq{})})&\leq \ex^{-\frac{m}{2}(\tau\rq{}-\ln(1+\tau\rq{}))}.
%\end{align*}
%Therefore, choosing $\tau'=3$, and letting $\sigma_m=\sigma\sqrt{1+\tau\rq{}}$ and $\e_m=\ex^{-0.5 m(\tau\rq{}-\ln(1+\tau\rq{}))}$ in Theorem \ref{Noisy Thm} yields the desired result.
%\end{proof}
%
%%--------------------------%------------------------------------%-----------------------

%---------------------%---------------------%---------------------%---------------------
\section{Information and complexity measures}\label{sec:RDD}

 %The signals in the universe that we care about and try to acquire are usually ``structured'', and do not look like arbitrary random noise. The field of compressed sensing is built upon the observation that exploiting this structure enables an acquisition system to accurately  estimate an $n$-dimensional  signal  from far fewer number of measurements than $n$. Sparsity, group-sparsity and low-rankness are some examples of the types of structure that are studied in the compressed sensing literature. 
 To develop a unified approach to the problem of structured signal recovery, and also to fundamentally understand the connections between the problems of data compression and compressed sensing, a universal notion of complexity for analog signals is required. Such a notion of complexity  is expected to effectively measure all  the information contained in the structure of an analog signal.

For discrete signals, there are well-known measures of complexity in the information theory literature. The entropy $H(X)$ and the entropy rate  $\bar{H}(\Xbbf)=\lim_{n\to\infty}H(X_n|X^{n-1})$ measure the complexity of random variable $X$ and  stationary  process $\Xbbf=\{X_i\}$, respectively. Both of these measures are closely connected to the minimum number of bits per symbol required for representing stochastic sources \cite{cover}. However, when we shift from discrete alphabet to analog, both the entropy, and the entropy rate become  infinite. Therefore, such measures cannot be used for capturing the structure of such signals.

To illustrate what is meant for an analog process to be structured, consider a stationary memoryless (i.e., i.i.d.) process $\Xbbf=\{X_i\}_{i=0}^{\infty}$ such that $X_i\sim (1-p)\delta_0+p f_c$, where $f_c$ denotes the probability density function (pdf) of an  absolutely continuous distribution. In other words, for each $i$, with probability $p$, $X_i$ is exactly equal to zero, otherwise, it is drawn from $f_c$. From this definition,  a block $X^n$ generated by this source contains around $n(1-p)$ entries equal to zero, and the rest of the entries are real numbers in  the domain of $f_c$. To describe $X^n$ with a certain precision, for zero entries, it suffices to describe their locations. The number of bits required for this description does not depend on the reconstruction quality. However, for the remaining approximately $np$ elements of $X^n$, it can be proved that the required number of bits grows proportionally to the desired reconstruction quality. This intuitively suggests that the probability $p$, which controls the number of non-zero elements in $X^n$, is a fundamental quantity related to the complexity of $X^n$.  This intuition is nicely captured by the notion of ID introduced by R\'enyi \cite{Renyi:59}.

%In order to characterize the fundamental limits of CSP algorithm applied to stochastic sources, in this section we study the RDD \cite{KawabataD:94} and its relationship to ID of stochastic ergodic processes.

\begin{definition}[R\'enyi information dimension \cite{Renyi:59}]
The R\'enyi  upper and lower IDs of  an analog random variable $X$ are defined as
\[
\bar{d}(X)=\limsup_{b\to\infty} {H(\langle X\rangle_b)\over \log b},
\]
and
\[
\underline{d}(X)=\liminf_{b\to\infty} {H(\langle X\rangle_b)\over \log b},
\]
respectively. If the two limits coincide, $d(X)=\bar{d}(X)=\underline{d}(X)$ is defined as the R\'enyi ID of $X$.
\end{definition}

Note that while the above definition of the R\'enyi IDs is in terms of the entropy of the $b$-level quantized version of $X$ normalized by the number of bits required for binary representation of it, $\log b$, it is easy to see that we can equivalently find them in terms of the entropy of the $b$-bit approximation of $X$, $[X]_b$, normalized by $b$, the number of bits i.e.
$
\bar{d}(X)=\limsup_{b\to\infty} {H([X]_b)\over b},
$
and
$
\underline{d}(X)=\liminf_{b\to\infty} {H([X]_b)\over b}.
$

The R\'enyi ID of a random variable serves as a measure of complexity for analog random variables.  To shed some light on this measure, consider the i.i.d. sparse source  $\Xbbf$ described earlier. It can be proved that the R\'enyi ID of each $X_i$ is equal to $p$, which is the probability that $X_i$ is non-zero \cite{Renyi:59}. Decreasing the parameter $p$ increases  the sparsity level of the output of such a source, and hence intuitively decreases its complexity. This phenomenon is captured by the R\'enyi ID of $X$.  In fact, $\delta_0$ can be changed to any discrete probability distribution and the result will not change since the R\'enyi ID of a discrete source is 0.    The notion of R\'enyi ID for random variables or vectors was extended in \cite{JalaliP:14-arxiv} to define the ID of analog stationary processes.

\begin{definition}[ID of a stationary process \cite{JalaliP:14-arxiv}]
  The  $k$-th order upper and lower IDs of   stationary process ${\Xbbf}=\{X_i\}_{i=-\infty}^{\infty}$ are defined as
   \[
  \bar{d}_k({\Xbbf})=\limsup_{b\to \infty} {1\over b}H([X_{k+1}]_b|[X^k]_b),
  \]
   and
   \[
   \underline{d}_k({\Xbbf})=\liminf_{b\to \infty} {1\over b}H([X_{k+1}]_b|[X^k]_b),
   \]
respectively. The upper and lower ID of process ${\Xbbf}$ are defined as
\[
\bar{d}_o({\Xbbf})=\lim_{k\to\infty}\bar{d}_k({\Xbbf})
 \]
 and
 \[
 \underline{ d}_o({\Xbbf})=\lim_{k\to\infty}\underline{d}_k(X),
  \]
  respectively, when the limits exist. If $\bar{d}_o({\Xbbf})=\underline{d}_o({\Xbbf})$, the ID of process $\mathbb{\mathbf X}$, ${d}_o({\Xbbf})$, is defined as ${d}_o({\Xbbf})=\bar{d}_o({\Xbbf})=\underline{d}_o({\Xbbf})$.
\end{definition}

For a stationary memoryless i.i.d. process ${\Xbbf}=\{X_i\}_{i=-\infty}^{\infty}$,  this definition coincides  with that of R\'enyi\rq{}s ID of the first order marginal distribution of the process $\Xbbf$.  That is $\bar{d}_o({\Xbbf})=\bar{d}(X_1)$ and $\underline{d}_o({\Xbbf})=\underline{d}(X_1)$. For sources with memory, taking the limit as   the memory parameter $k$ grows to infinity allows $d_o(\Xbbf)$ to capture the overall structure that is  present  in an analog stationary process. It can be proved that $d_o(\Xbbf)\leq 1$, for all stationary processes, and if the stationary process $\Xbbf$ is structured,  $d_o(\Xbbf)$ is strictly smaller than one \cite{JalaliP:14-arxiv}. As an example of a structured stationary analog process with memory, consider a piecewise constant signal modeled by a first order Markov process ${\Xbbf}=\{\X_\indx\}_{\indx=1}^\infty$, such that  conditioned on $\X_{\indx-1}=\x_{\indx-1}$, $\X_\indx$ is distributed according to $(1-p)\delta_{\x_{\indx-1}}+pf_c$ where $f_c$ denotes the pdf of an absolutely continuous distribution with bounded support, defined over an interval $(l,u)$. In other words, at each time $i$,  the process either makes a jump and takes a value drawn from distribution $f_c$, or it stays at $X_{i-1}$. The decision is made based on the outcome of an i.i.d. $\operatorname{Bern} \left({p}\right)$ random variable independent of all past values of $\Xbbf$. While the output of this source is  not  sparse,  it is clearly a structured process. This intuition is indeed captured by the ID of the process;  it can be proved that ${d}_o(\Xbbf)=p$, \ie the probability that the process makes a jump determines the complexity of this process \cite{JalaliP:14-arxiv}. 

%In other words, at each time $i$, independent of the previous values, it either makes a jump and takes a value drawn from distribution $f_c$, or it stays at $X_{i-1}$. While the output of this source is  not  sparse,  the source is clearly a structured process. It can be proved that for ${d}_o(\Xbbf)=p$, \ie the probability that the process makes a jump determines the complexity of this process \cite{JalaliP:14-arxiv}. 

%\begin{lemma}{(Lemma 2 in \cite{JalaliP:14-arxiv})}\label{lemma:upper-ID-alter}
%The upper ID of a stationary process ${\Xbbf}$ can be written as
%\[
%\bar{d}_o({\Xbbf})=\lim_{k\to\infty}{1\over k}\Big(\limsup_{b\to\infty} {H([X^k]_b)\over b}\Big).
%\]
%\end{lemma}

For a stationary memoryless process, under some mild conditions on the distribution, \cite{WuV:10} proves that the R\'enyi ID of the first order marginal distribution of the source characterizes the fundamental limits of compressed sensing. In other words, given a process ${\Xbbf}$,  asymptotically, as the blocklength grows to infinity,  the minimum number of linear projections, $m$, normalized by the ambient dimension, $n$, that is required for  recovering source $X^n$ from its linear projections is shown to be equal to $d(X_1)$, which is the R\'enyi ID of $X_1$. In \cite{JalaliP:14-arxiv}, it is shown that asymptotically slightly more than $n\bar{d}_o({\Xbbf})$ random linear projections suffice  for \emph{universal} recovery of $X^n$ generated by any Markov process of any order, without knowing the source model, where $\bar{d}_o({\Xbbf})$ denotes the upper ID of the process $\Xbbf$. These results provide an operational interpretation to the R\'enyi ID of a random variable and its generalization to stationary processes.

The focus of this paper is on the application of compression codes in building compressed sensing recovery algorithms. The rate-distortion function of a stationary source measures the minimum number of bits per source symbol required for  achieving a given reconstruction quality. It turns out that for an analog process as the reconstruction becomes finer, the behavior the rate-distortion function is connected to the level of structuredness of the source process and ID notions mentioned earlier. In the rest of this section, we first review the known results on this connection, and then  prove our main result of this section, which, under some mild conditions,  establishes this connection for general stationary processes.

Consider a metric space  $(\mathds{R}^k,\rho)$, and random vector $X^k$. The rate-distortion function of $X^k$ under expected distortion constraint
\[
d(x^k,\xh^k)=\rho(x^k,\xh^k)^r
\]
is defined as
\[
R_r(X^k,D)=\inf_{ \E[d(X^k,\Xh^k)]\leq D}I(X^k;\Xh^k).
\]
\begin{definition}[Rate-distortion dimension (RDD) of a random vector \cite{KawabataD:94}]
The upper and lower RDDs of $X^k$ are defined as
\[
\overline{\dim}_R(X^k)=r\limsup_{D\to0}{R_r(X^k,D)\over \log{1\over D}},
\]
and
\[
\underline{\dim}_R(X^k)=r\liminf_{D\to0}{R_r(X^k,D)\over \log{1\over D}},
\]
 respectively. If $\overline{\dim}_R(X^k)=\underline{\dim}_R(X^k)$, the RDD of $X^n$ is defined as ${\dim}_R(X^k)=r\lim_{D\to0}{R_r(X^k,D)\over \log{1\over D}}$.
\end{definition}

The following theorem from \cite{KawabataD:94} establishes the connection between the R\'enyi ID of a random vector $X^k$ and its RDD, for any general distribution on $X^k$.

\begin{theorem}[Proposition 3.3 in \cite{KawabataD:94}] \label{thm:prop3-3}
Consider the metric space $(\mathds{R}^k,\rho)$, such that there exists $0<a_1\leq a_2<\infty$ for  which $a_1\max_{i=1}^k|x_i-\xh_i|\leq \rho(x^k,\xh^k)\leq a_2\max_{i=1}^k|x_i-\xh_i|,$
for all $x^k,\xh^k\in\mathds{R}^k$. Then,  for any distribution of $X^k$,
\[
\overline{\dim}_R(X^k)=\bar{d}(X^k),
\]
and
\[
\underline{\dim}_R(X^k)=\underline{d}(X^k),
\]
where $\overline{\dim}_R(X^k)$,  and $\underline{\dim}_R(X^k)$ denote the upper and lower RDD of $X^k$ under fidelity constraint $d(x^k,\xh^k)=\rho(x^k,\xh^k)^r$.
\end{theorem}

%--------------------------%------------------------------------%-----------------------

%\begin{remark}For an i.i.d.~process, under some mild conditions on the distribution, \cite{WuV:10} proves that R\'enyi ID of $X_1$ characterizes the fundamental limits of compressed sensing. In other words,  the minimum number of linear projections, $m$, normalized by the ambient dimension, $n$, that is required for  recovering source ${\Xbbf}$ from its linear projections is shown to be equal to $d(X_1)$. In \cite{JalaliP:14-arxiv}, it is shown that asymptotically slightly more than $2n\bar{d}_o({\Xbbf})$ random linear projections suffice  for \emph{universal} recovery of $X^n$ generated by any Markov process of any order, without knowing the source model.
%\end{remark}
%---------------------%---------------------%---------------------%---------------------

Consider an analog stationary process ${\Xbbf}=\{X_i\}_{i=-\infty}^{\infty}$. The rate-distortion function $R({\Xbbf},D)$ of the source ${\Xbbf}$ under squared error distortion can be computed as \cite{book:Berger,Gallager}
\[
R({\Xbbf},D)=\lim_{m\to\infty} R^{(m)}({\Xbbf},D),
\]
where
\[
R^{(m)}({\Xbbf},D)=\inf_{\E[d_m(X^m,\Xh^m)]\leq D}{1\over m} I(X^m;\Xh^m).
\]
and
\begin{align}
d_m(x^m,\xh^m)={1\over m}\|x^m-\xh^m\|_2^2. \label{sq-err-distortion}
\end{align}

Note that with this distortion metric, we have $r=2$ and $R^{(m)}({\Xbbf},D)= {1\over m}R_2({X}^m,D)$. It can also be shown that $\inf_{m} R^{(m)}({\Xbbf},D)=R({\Xbbf},D)$ \cite{Gallager}. %We extend the RDD from random vectors to stationary stochastic processes as an operational measure of complexity for such sources.

\begin{definition}[RDD of a stationary process]
The upper and lower RDDs of this stationary process $\Xbbf$ can be defined as
\[
\overline{\dim}_R({\Xbbf})=2\limsup_{D\to0}{R({\Xbbf},D)\over \log{1\over D}}
\]
and
\[
\underline{\dim}_R({\Xbbf})=2\liminf_{D\to0}{R({\Xbbf},D)\over \log{1\over D}}.
\]
If $\overline{\dim}_{R}(\Xbbf)=\underline{\dim}_{R}(\Xbbf)$, then $\dim_{R}(\Xbbf)=\overline{\dim}_{R}(\Xbbf)=\underline{\dim}_{R}(\Xbbf)$ is the RDD of $\Xbbf$.
\end{definition}
The main result of this section is the following theorem  which extends the equivalence of R\'enyi ID and RDD shown in \cite{KawabataD:94} for i.i.d. random vectors to stationary processes.

\begin{theorem}\label{thm:ID-eq-RDD}
For a stationary process $\Xbbf=\{X_i\}_{i=-\infty}^{\infty}$, assume that $\lim_{D\to 0}  { R^{(m)}({\Xbbf},D)\over \log {1\over D}}$ exists for all $m$. Then,
\[
{\dim}_{R}(\Xbbf) = \bar{d}_o({\Xbbf}).
\]
\end{theorem}

The main ingredients of the proof of  Theorem \ref{thm:ID-eq-RDD} are the following two lemmas.

%--------------------%--------------------%--------------------%--------------------
%--------------------%--------------------%--------------------%--------------------

\begin{lemma}\label{lemma:connect-UID-URDD}
For any stationary process $\Xbbf$, we have
\[
\overline{\dim}_{R}(\Xbbf) \leq \bar{d}_o({\Xbbf})\leq \inf_{m} 2\Big(\limsup_{D\to0}  { R^{(m)}({\Xbbf},D)\over \log {1\over D}}\Big).
\]
\end{lemma}

%--------------------%--------------------%--------------------%--------------------
%--------------------%--------------------%--------------------%--------------------
\begin{lemma}\label{lemma:uniform_conv}
Assume that $\lim_{D\to 0}  { R^{(m)}({\Xbbf},D)\over \log {1\over D}}$ exists for all $m$, and also there exists $\sigma_{\max}^2>0$, such that ${R^{(m)}({\Xbbf},D)}$ uniformly converges  to ${R({\Xbbf},D)}$, for $D\in(0,\sigma_{\max}^2)$, as $m$ grows to infinity. Then, ${\dim}_{R}(\Xbbf) = \bar{d}_o({\Xbbf}).$

\end{lemma}

Proofs of  Theorem \ref{thm:ID-eq-RDD} and Lemmas \ref{lemma:connect-UID-URDD} and \ref{lemma:uniform_conv} are provided in Section \ref{sec:proofs}.

%________________________________________________________________%

To illustrate the relationship between RDD and ID, as an example, consider the piecewise-constant signal described earlier. To directly evaluate the RDD of this process,  its rate-distortion characterization is required. However, deriving the rate-distortion function of sources with memory is in general very challenging. For instance, even for the binary symmetric Markov chain, the rate-distortion function is not known, except in a low-distortion region \cite{Gray:71}, and we have to resort to upper and lower bounds in general \cite{Berger:77,JalaliW:07}. The following theorem provides upper and lower bounds on the $R(\mathbf{X},D)$ of the piecewise-constant source. While there is a gap between the bounds on the $R(\mathbf{X},D)$, since the gap does not depend on $D$, as shown in the following corollary, they can be used to evaluate the RDD of the source exactly.

\begin{theorem}\label{thm:r-d-bound}
Consider a first-order stationary Markov process ${\Xbbf}=\{\X_\indx\}_{\indx=0}^\infty$, such that  conditioned on $\X_{\indx-1}=\x_{\indx-1}$, $\X_\indx$ is distributed according to $(1-p)\delta_{\x_{\indx-1}}+pf_c$, where $f_c$ denotes the pdf of an absolutely continuous distribution with bounded support,  $(l,u)$. 
 If $d_{\max}\triangleq \sup_{x,\xh\in(l,u)}d(x,\xh)<\infty$, then
\[
pR_{f_c}(D)\leq R(\mathbf{X},D)\leq H(p)+pR_{f_c}(D),
\]
  where $R_{f_c}(D)$  and $H(p)$  denote the rate distortion function of an i.i.d.~process distributed according to pdf $f_c$, and the binary entropy function ($-p\log_2 p-(1-p)\log_2(1-p)$), respectively.
\end{theorem}

\begin{proof}
A detailed proof of Theorem \ref{thm:r-d-bound} is presented in Section \ref{sec:proofs}. To prove the upper bound (achievability), we consider a code that describes  the positions of the jumps  losslessly at rate $H(p)$. Since the source is piecewise constant, after describing the positions of the jumps,  the encoder removes the repeated values and applies a lossy compression code of blocklength length  close to $np$. Therefore, to describe the values at distortion $D$ the encoder roughly needs to spend $npR_{f_c}(D)$ bits. %While the formal proof follows this intuition, there are some technical intricacies in dealing with the randomness of this shortened blocklength.
For the lower bound (converse), we consider a genie-aided  decoder that has access to the positions of the jumps.  Then intuitively, to describe the values at distortion $D$, it still needs a rate of at least $pR_{f_c}(D)$. %While this  gives the outline of the converse bound, the actual proof contains a lot of details, since both the distortion and the number of samples coming from $f_c$ depend on the reduced block length which is, as mentioned before, a random number.
The proof in Section \ref{sec:proofs} makes these steps formal by properly analyzing the reduced block length which is a random number.
%to lower bound the rate conditioned on this extra information and carefully go around the convexity arguments in the converse proof of the classic rate-distortion, since both the distortion and the number of samples coming from $f_c$ depend on a the reduced block length which is, as mentioned before, a random number.
\end{proof}

\begin{corollary}\label{thm:RDD-piecewise-constant}
%The almost lossless performance of the CSP algorithm when applied to the piecewise constant source described earlier depends on the RDD of this source.  By Theorem 2 in \cite{JalaliP:14-arxiv},
% \[
% \bar{d}_o(\Xv)=\underline{d}_o(\Xv)=p.
% \]
For the piecewise constant source in Theorem \ref{thm:r-d-bound}, %if the condition of  Theorem \ref{thm:ID-eq-RDD} holds, then 
%the RDD is equal to $p$ which is in turn equal to the ID of this source:
we have
 \[
 {\dim}_R({\Xbbf})=\bar{d}_o(\Xv)=p.
 \]
 In other words, the RDD is equal to $p$ which is in turn equal to the ID of this source.
 \end{corollary}
 \begin{proof}
 Given the bound on the rate-distortion process derived in Theorem \ref{thm:r-d-bound}, it is easy to directly derive the  RDD of such a source. More precisely, given the upper bound, it follows that
 \begin{align*}
 \overline{\dim}_R({\Xbbf})&=2\limsup_{D\to0}{R({\Xbbf},D)\over \log{1\over D}}\\
 &\leq 2\limsup_{D\to0} { H(p)+pR_{f_c}(D) \over \log{1\over D}}\\
 &= p(\limsup_{D\to0} {R_{f_c}(D) \over \log{1\over D}})\\
 &=p.
 \end{align*}
Similarly, given the lower bound, we have
\begin{align*}
 \underline{\dim}_R({\Xbbf})&=2\liminf_{D\to0}{R({\Xbbf},D)\over \log{1\over D}}\\
 &\geq 2\liminf_{D\to0} {pR_{f_c}(D) \over \log{1\over D}}\\
 &= p(\liminf_{D\to0} {R_{f_c}(D) \over \log{1\over D}})\\
 &=p,
 \end{align*}
 where the last lines in both the upper and the lower RDDs follow from \cite{KawabataD:94} and \cite{Renyi:59}. Therefore, $p\leq\underline{\dim}_R({\Xbbf})\leq\overline{\dim}_R({\Xbbf})\leq p$. In other words, for this source RDD exists and is equal to ${\dim}_R({\Xbbf})=p$. Hence, the condition of  Theorem \ref{thm:ID-eq-RDD} holds and we have
\[
 {\dim}_R({\Xbbf})=\bar{d}_o(\Xv).
 \]
 This agrees with the ID of this source found in Theorem 2 in \cite{JalaliP:14-arxiv},
 \[
 \bar{d}_o(\Xv)=\underline{d}_o(\Xv)=p.
\]
\end{proof}

% Therefore, by Lemma \ref{lemma:connect-UID-URDD},
% \[
% \overline{\dim}_{R}(\Xbbf) \leq p.
% \]
% On the other hand, from Theorem \ref{thm:r-d-bound},
% \begin{align*}
% \overline{\dim}_R({\Xbbf})&=2\limsup_{D\to0}{R({\Xbbf},D)\over \log{1\over D}}\\
% &\leq 2\limsup_{D\to0} { H(p)+pR_{f_c}(D) \over \log{1\over D}}\\
% &= p(\limsup_{D\to0} {R_{f_c}(D) \over \log{1\over D}})\\
% &=p,
% \end{align*}
% where the last line follows from \cite{KawabataD:94} and \cite{Renyi:59}. It can be observed  that the latter upper bound complies by the former bound derived from Lemma \ref{lemma:connect-UID-URDD}.

\begin{remark}
Corollary \ref{thm:RDD-piecewise-constant} states that the RDD of the piecewise constant source described in Theorem \ref{thm:r-d-bound} is equal to $p$, which is also the ID of this process \cite{JalaliP:14-arxiv}. While \cite{JalaliP:14-arxiv} directly computes the ID of such processes, Theorem \ref{thm:ID-eq-RDD}, by proving  the equivalence of ID and RDD, provides  a potentially easier alternative path to computing the ID of stochastic  processes.  Note that to be able to calculate the RDD of a process, the exact characterization of the rate-distortion function is not required. In fact, it is easy to see that it would be enough to have upper and lower bounds on  the rate-distortion function of the source, $R(\Xbbf,D)$, that are within a reasonable gap. More precisely,  as long as the gap between the bounds grows as $o(\log{1 \over D})$, they can be used to evaluate the RDD. 
%\textcolor{red}{
Moreover, since the RDD only depends on the low-distortion behavior  of the rate-distortion function,  studying its asymptotic small distortion performance  is sufficient for computing the RDD and  as a result the ID of a source, without knowing the rate-distortion function explicitly.  For instance,  \cite{Gyorgy99} studies the asymptotic behavior of the rate-distortion function of  some stochastic sources and employs those results to evaluate the RDD of some i.i.d.~processes.
%}
%In summary, studying the  rate-distortion behavior of a process  provides an alternative  approach to computing  its ID.
\end{remark}

%------------------%------------------
%------------------%------------------
%------------------%------------------

%--------------------%--------------------%--------------------%--------------------
\section{Almost lossless recovery}\label{sec:AlmostLosslessCSP}
%----------------------------------------------------------------------------------------------------%

Section \ref{sec:CSP-performance} formulated the performance of the CSP algorithm which employs  a lossy compression code to recover the output of a stationary process from random linear projections. Specifically, Theorem \ref{NoiselessThm} and Theorem \ref{Noisy Thm} characterize the performance of the CSP algorithm, for noiseless and noisy measurements, respectively. In this section, we focus on the special case in which the lossy compression code  is a high-resolution one, and therefore, $D$ is very small. As a result, with high probability, the CSP algorithm generates a high-fidelity or almost lossless reconstruction  of the input vector. While the CSP algorithm is inherently a lossy CS recovery algorithm due to the utilized lossy compression code, the almost lossless recovery performance can be achieved by letting the distortion of the compression code become arbitrarily small.  To build the analytical tools and insights required to evaluate the CSP performance when the distortion approaches zero,  in Section \ref{sec:RDD}, we focused on measures of complexity for stationary analog processes, and established a connection with the RDD of a stationary process and its ID. %In this section, we focus on the special case in which the lossy compression code  is a high-resolution code, and therefore, $D$ is very small. As a result, with high probability, the CSP algorithm generates a high-fidelity or almost lossless reconstruction  of the input vector. 

%%%%%%%%%%%%%%%%%%%%%%%%%%
% In a standard compressed sensing setting, the decoder recovers the signal \emph{losslessly} or almost losslessly from an  underdetermined set of linear equations. While compression-based recovery algorithms enable  us to exploit more complex structures, there is an inherent loss due to the underlying lossy compression codes. By letting the distortion of the compression code become arbitrarily small, we can achieve almost lossless recovery which is of interest in compressed sensing problems. Although arbitrarily small distortion for an analog signal dictates an arbitrarily large compression code rate, the RDD of the code,  which is the quantity that directly relates the compression code's rate-distortion behavior to the number of required measurements, remains bounded.
%%%%%%%%%%%%%%%%%%%%%%%%%%

In a noiseless setting, Theorem \ref{NoiselessThm} asserted that given a compression code operating at rate $R$ and distortion $D$, the CSP algorithm is able to recover signal $X^n$ from ${2\eta R \over \log({1 \over D})}$ randomized linear measurements. Note that the RDD of the source was defined in Section \ref{sec:RDD} as $\lim_{D\to 0}{2 R(\Xbbf,D) \over \log({1 \over D})}$. Therefore, considering a family of optimal compression codes that operate at a very low distortion level, Theorem \ref{NoiselessThm} predicts that, asymptotically, if the normalized number of measurements is slightly higher than the RDD of the source, then the CSP algorithm generates an almost lossless reconstruction. This result is formalized in the  following corollary, which  studies the performance of the CSP algorithm in the extreme case, where $D$ approaches zero. It proves that as long as the normalized number of measurements is larger than $\overline{\dim}_R({\Xbbf})$, CSP recovers the source vector almost losslessly.

%Assuming that $\overline{\dim}_{R}(\Xbbf)$ and $\underline{\dim}_{R}(\Xbbf)$ are both bounded,
\begin{corollary} \label{CSP-RDD-Corollary}
Consider a stationary process $\Xbbf$ and a system of random linear observations, $Y^m=A X^n$,  with measurement matrix $A\in {\rm I\!R}^{m\times n}$, where $A_{i,j}$ are   i.i.d.~as $\mathcal{N}(0,1)$. For any observation error $\Delta>0$, if the number of measurements $m=m_n$ satisfies
\[
\liminf_{n\to\infty}{m_n\over n}> \overline{\dim}_R({\Xbbf}),
 \]
 then there exists a family of compression codes which, when used by the CSP algorithm, yields
\[
\lim_{n\to\infty}\P(\frac{1}{\sqrt{n}}\lVert\; X^n - \Xt^n\rVert_2 \geq \Delta )\to 0,
\]
where $\Xt^n$ refers to the solution of the CSP algorithm as in \eqref{CSPsolution}.
\end{corollary}

\begin{proof}
The proof is provided in Section \ref{sec:proofs}.
\end{proof}

%\begin{remark}
Applying Corollary \ref{CSP-RDD-Corollary} to the piecewise constant source described in Theorem \ref{thm:r-d-bound}, implies that the almost lossless performance of the CSP algorithm for such a source depends on the RDD of this source. Corollary \ref{thm:RDD-piecewise-constant} shows that the for the piecewise constant source we have 
 $
 {\dim}_R({\Xbbf})=\bar{d}_o(\Xv)=p.
 $
 Combining  this together with Corollary \ref{CSP-RDD-Corollary} implies  that there exists a family of compression codes that when employed by the CSP algorithm  with a number of measurements satisfying  $
\liminf_{n\to\infty}{m_n\over n}> p$  yields an asymptotic  almost lossless recovery of this source.
 
%\end{remark}

\begin{remark}
Corollary \ref{CSP-RDD-Corollary} states that the CSP algorithm can achieve almost lossless recovery, using slightly more than $n\overline{\dim}_R({\Xbbf})$ random linear measurements. On the other hand, for i.i.d. sources, under some mild conditions,   $n{d}_o(\Xbbf)$ characterizes the  minimum required number of measurements for almost lossless recovery \cite{WuV:10}. Note that if the rate-distortion function of the source satisfies the condition of  Theorem \ref{thm:ID-eq-RDD}, then  ${\dim}_R({\Xbbf})=\bar{d}_o({\Xbbf})$. Even without such assumption on the rate-distortion function, we can employ Lemma \ref{lemma:connect-UID-URDD} to upper bound $\overline{\dim}_{R}(\Xbbf)$ by $\bar{d}_o({\Xbbf})$ and get the same result. Therefore, at least for memoryless i.i.d. sources, the CSP algorithm achieves the optimal performance, in terms of achieving the minimum number of measurements.
\end{remark}

For general stationary sources, if the decoder is restricted to be Lipschitz-continuous, which is formally defined below, then asymptotically the normalized number of measurements should be larger than $\limsup_{n\to\infty} {\bar{d}(X^n)/n}$ \cite{WuV:12}, which is equal to $\bar{d}_o(\Xbbf)$ \cite{JalaliP:14-arxiv}. While the CSP decoder is not Lipschitz continuous, we conjecture that the lower bound also holds for less-restricted decoders. Proving or disproving this is an interesting topic for future research. 

\begin{definition}[Lipschitz continuity]
Consider a set $\Ac\subset\mathds{R}^k$. A function $f:\Ac\to\mathds{R}^n$ is called Lipschitz continuous if there exists constant $c\in\mathds{R}$, such that
\[
\|f(x)-f(y)\|\leq c\|x-y\|,
\]
for all $x,y\in\Ac$.
\end{definition}

%\begin{proposition}[Proposition 2.5 in \cite{Falconer:book}]\label{prop:box-counting-dim}
%Consider $\Ac\subset\mathds{R}^n$ and function $f:\Ac\to\mathds{R}^m$.
%  \begin{enumerate}
%    \item[(a)] If $f$ is a Lipschitz function, that is, there exists $c>0$, such that $|f(x)-f(y)|\leq c|x-y|$, for all $x,y\in\Ac$, then
%        \[
%        \overline{\dim}_{\rm B}(f(\Ac))\leq \overline{\dim}_{\rm B}(\Ac)
%        \]
%        and $\underline{\dim}_{\rm B}(f(\Ac))\leq \underline{\dim}_{\rm B}(\Ac).$
%    \item[(b)] If $f$ is a bi-Lipschitz function, that is, there exist $c_1,c_2>0$, such that $c_1|x-y|\leq |f(x)-f(y)|\leq c_2|x-y|$, for all $x,y\in\Ac$, then
%        \[
%        \overline{\dim}_{\rm B}(f(\Ac))=\overline{\dim}_{\rm B}(\Ac),
%        \]
%        and $\underline{\dim}_{\rm B}(f(\Ac))=\underline{\dim}_{\rm B}(\Ac).$
%  \end{enumerate}
%\end{proposition}

%----------------------------------------------------------------------------------------------------%

\section{From universal compression to universal compressed sensing}\label{sec:UCSP-performance}

In compressed sensing,  the decoder tries to find the signal that matches the measurements and also has the same structure as the unknown input signal. In many applications, the structure of the input signal is  not known by the decoder or is known only partially. For instance, for an image, the decoder might know that the wavelet coefficients of the image are sparse, but the image might have much more structure that is not known by the decoder.  Moreover, in many application, it is desired to have decoders  that work well for sources with different statistics. In summary, from a practical viewpoint, it appealing to have  decoders that take advantage of all the information contained in the structure of the signal,  without having any prior knowledge about the source distribution, Such decoders, potentially, lead to very efficient compressed sensing algorithms that work for various source models.

A universal compressed sensing decoder aims at recovering an input signal from its under-determined linear measurements, without having access to the source distribution or the source model.  The existence of such universal recovery algorithms is known for both deterministic \cite{JalaliM:14} and stochastic \cite{BaronD:11,BaronD:12,JalaliP:14-arxiv} settings. In this section, we prove that a family of universal compression codes combined by the CSP algorithm leads to a family of universal compressed sensing recovery algorithms.

Consider a family of variable-length point-wise universal lossy compression codes $(n,f_n,g_n)$ for analog stationary ergodic processes with alphabet $\Xc\subset\mathds{R}$. Assume that the family of codes $(n,f_n,g_n)$ operates at fixed distortion $D$.  That is,  for any stationary ergodic process $\Xbbf=\{X_i\}_i$, with $X_i\in\Xc$,
\begin{enumerate}
  \item[i)] $\lim_{n\to\infty}{1\over n}|f_n(X^n)| = R(\Xbbf,D)$, almost surely,
  \item[ii)] $\lim_{n\to\infty}\P({1\over n}\|X^n-\Xh^n\|^2_2\geq D+\e)=0$,
\end{enumerate}
for any $\e>0$.

Consider $X^n$ generated by a stationary ergodic process $\Xbbf$ with rate-distortion function $R(\Xbbf,D)$. A universal compressed sensing  decoder observes $Y^m=AX^n$, and aims at estimating $X^n$ from $Y^m$, employing the code $(f_n,g_n)$, without having access to the distribution of the source. To achieve this goal, consider the following slightly modified version of the CSP algorithm, which we refer to as universal CSP (UCSP):
\begin{align}\label{eq:universal-csp}
  \min &\;\;\;\; \|Au^n-Y^m\|_2\nonumber \\
  {\rm s.t.} &\;\;\;\;u^n=g_n(b),\nonumber \\
  &\;\;\;\; b\in\{0,1\}^*, |b|\leq n( R(\Xbbf,D)+\e).
\end{align}
In other words,  among all binary sequences of length smaller than $n(R(\Xbbf,D)+\e)$, UCSP searches for the one whose decompressed version via the universal  decoder $g_n$ yields the smallest measurement error.

%Given $X^n$ generated by stationary ergodic process $X$, and measurements $Y^m=AX^n$, 
For $A_{ij}\stackrel{\rm i.i.d.}{\sim}\Nc(0,1)$, let $\Xt^n$ denote the minimizer of the UCSP algorithm that employs a point-wise universal compression code operating at distortion $D$. The following theorem characterizes the performance of the UCSP algorithm and proves that a universal compression code leads to a universal compressed sensing algorithm.

\begin{theorem}\label{thm:universal-CS}
Consider $Y^m=A X^n$, a system of random linear observations  with measurement matrix $A\in {\rm I\!R}^{m\times n}$, where $A_{i,j}\sim$  i.i.d $\mathcal{N}(0,1)$. Let $\Cc_n$ be variable-length point-wise universal lossy compression code operating at rate $R$ that achieves  distortion $D$ with excess distortion probability $\epsilon$.
%Without any loss of generality assume that the source is normalized such that $D<1$. 
 For   $\a>0$, $\e>0$, and $\eta>1$, such that ${\eta\over \log{1\over D}}+\a>\e$, let $\d={\eta\over \log{1\over D}}+\a-\e$. Suppose $\Xt^n$ refers to the solution of the UCSP algorithm as in \eqref{eq:universal-csp}. Then,  for $n$ large enough, and
 \[
{m\over n} = \frac{2\eta R(\Xbbf,D)}{\log{1\over D}},
\] we have
\begin{align*}
\P&\Big(\frac{1}{\sqrt{n}}\lVert\; X^n - \Xt^n\rVert_2 \geq (2+\sqrt{n/m})D^{\frac{1}{2}(1-\frac{1+\d}{\eta})}  \Big)\\
&\;\;\leq \e+2^{-\frac{1}{2}nR\alpha}+ \mathrm{e}^{-{m\over 2}}.
\end{align*}

\end{theorem}

\begin{proof}
The proof is provided in Section \ref{sec:proofs}.
\end{proof}

Comparing Theorem \ref{thm:universal-CS} with Theorem \ref{NoiselessThm}, it can be observed that the performance trade-offs for the CSP and the UCSP are exactly the same in terms of the rate-distortion behavior of the underlying compression code. The difference between the two is in the fact that   the CSP optimization employs a compression code that is designed for input source distribution, but the UCSP optimization requires a universal compression code.   This might suggest that since UCSP has then same asymptotic performance as CSP, and in addition works for any input distribution, it is always a better choice than the CSP. However, note that UCSP is build upon  point-wise universal compression codes for analog sources. While such codes theoretically exist, practical instances of such codes are yet to be found. Moreover, another potential disadvantage of the UCSP compared to the CSP optimization is that while universal codes usually achieve the same asymptotic performance as non-universal codes, their finite blocklength performance is worse than non-universal codes.  

Similar to Corollary \ref{CSP-RDD-Corollary}, the following corollary considers the special  case where the distortion approaches zero, and proves that, as long as the normalized number of measurements is larger than $\overline{\dim}_R({\Xbbf})$, there exist universal compression codes that yield  universal compressed sensing algorithms that can estimate the source almost losslessly. Note that $\overline{\dim}_R({\Xbbf})$ is the RDD of the source $\Xbbf$, which depends on the source model and captures all the structure within the signal.

%Assuming that $\overline{\dim}_{R}(\Xbbf)$ and $\underline{\dim}_{R}(\Xbbf)$ are both bounded,
\begin{corollary} \label{cor-UCSP}
Consider a stationary process $\Xbbf$ and a system of random linear observations, $Y^m=A X^n$,  with measurement matrix $A\in {\rm I\!R}^{m\times n}$, where $A_{i,j}$ are  i.i.d.~  as $\mathcal{N}(0,1)$ and $m=m_n$ is the number of observations. For any observation error $\Delta>0$, if the sequence $m_n$ satisfies
\[
\liminf_{n\to\infty}{m_n\over n}> \overline{\dim}_R({\Xbbf}),
 \]
then there exists a family of variable-length point-wise universal lossy compression codes which, when used by the UCSP algorithm, yields
\[
\lim_{n\to\infty}\P\left(\frac{1}{\sqrt{n}}\lVert\; X^n - \Xt^n\rVert_2 \geq \Delta \right)\to 0.
\]
where $\Xt^n$ refers to the solution of the UCSP algorithm as in \eqref{eq:universal-csp}.

%There exists a family of universal  lossy compression codes, such that when it is used by the CSP algorithm,  for any stationary process $\Xbbf$, and any $\Delta>0$, if $m=m_n$ satisfies
%\[
%\liminf_{n\to\infty}{m_n\over n}> \overline{\dim}_R({\Xbbf}),
% \]
% then the output of the UCSP algorithm satisfies
%\[
%\lim_{n\to\infty}\P\left(\frac{1}{\sqrt{n}}\lVert\; X^n - \Xt^n\rVert_2 \geq \Delta \right)\to 0.
%\]
\end{corollary}
\begin{proof}
The proof is very similar to the proof of Corollary \ref{CSP-RDD-Corollary} and is omitted.
\end{proof}

\section{Proofs}\label{sec:proofs}

The following lemma from \cite{JalaliM:13-isit} is used in some of the proofs.

\begin{lemma}[$\chi^2$-construction]\label{chi2}
   Fix $\tau>0$ and let $\Z_i\sim\mathcal{N}(0,1)$, $i=1,2,\dots,m$. Then,
   \begin{IEEEeqnarray*}{rl}
       \P\big(\sum_{i=1}^m\Z_i^2<m(1-\tau)\Big)\leq\mathrm{e}^{\frac{m}{2}(\tau+\ln(1-\tau))}
   \end{IEEEeqnarray*}
   and
   \begin{IEEEeqnarray*}{rl}
       \P\Big(\sum_{i=1}^m\Z_i^2>m(1+\tau)\Big)\leq\mathrm{e}^{-\frac{m}{2}(\tau-\ln(1+\tau))}.
   \end{IEEEeqnarray*}
\end{lemma}

\subsection{Proof of Theorem \ref{NoiselessThm}}\label{Noiseless-Proof}
%To prove Theorem \ref{NoiselessThm},
Let $\hat{\X}^n=g_n(f_n(X^n))$. Since  $\tilde{\X}^n=\argmin_{x^n\in\Cc_n}\|Y^m-Ax^n\|_2^2$, and $\Xh^n\in\Cc_n$,
\[
\|Y^m-A\tilde{\X}^n\|_2\leq\|Y^m-A\hat{\X}^n\|_2.
\] Substituting $AX^n$ for $Y^m$, it follows that
\begin{align}\label{eq:triangle-ineq}
   \|A(X^n-\tilde{\X}^n)\|_2\leq\|A(X^n-\hat{\X}^n)\|_2.
\end{align}
Define the event $\mathcal{E}_0$ as
\begin{align*}
\Ec_0&\triangleq \{\|X^n-\hat{\X}^n\|_2^2\leq nD\}.
\end{align*}
By assumption, $\P(\Ec_0^c)\leq \e$. Conditioned on $\mathcal{E}_0$, from \eqref{eq:triangle-ineq}, we have
\begin{align}
   \|A(X^n-\tilde{\X}^n)\|_2&\leq \sigma_{\max}(A)\sqrt{nD}\label{eq:bd-sigma-max}
\end{align}
where $\sigma_{\max}(A)$ is the maximum singular value of $A$.
 Define events $\Ec_1$ and $\Ec_2$ as
\begin{align*}
   \mathcal{E}_1\triangleq\big\{&\forall \xt^n\in\Cc:\\
   &\;\;\|A(\X^n-\xt^n)\|_2\geq\sqrt{(1-\tau)m}\|\X^n-\xt^n\|_2\big\},
\end{align*}
   where $\tau\in(0,1)$,
   and
   \[
   \mathcal{E}_2\triangleq\{\sigma_{\max}(A)-\sqrt{m}-\sqrt{n}<\sqrt{m} \}.
\]
Then, conditioned on  $\Ec_0\cap\Ec_1\cap\Ec_2$, it follows from \eqref{eq:bd-sigma-max} that
\begin{align}
   \sqrt{m(1-\tau)}\|X^n-\tilde{\X}^n\|_2&\leq\|A(X^n-\tilde{\X}^n)\|_2\nonumber\\
   &\leq\sigma_{\max}(A)\sqrt{nD}\nonumber\\
   &\leq(\sqrt{n}+2\sqrt{m})\sqrt{nD}.\label{eq:cond-on-E1-2-3}
\end{align}
Rearranging the terms and setting $m=\frac{2\eta nR}{\log(1/D)}$ and $\tau=1-D^{(1+\d)/\eta}$ in \eqref{eq:cond-on-E1-2-3} yields
\begin{align}
   \frac{1}{\sqrt{n}}\|X^n-\tilde{\X}^n\|_2 &\leq \sqrt{D \over 1-\tau}\left(\sqrt{\frac{n}{m}}+2\right)\\
 =& \sqrt{\frac{D}{D^{(1+\d)/\eta}}}\left(\sqrt{\frac{n}{m}}+2\right)\nonumber\\
 =& D^{0.5(1-(1+\d)/\eta)}\left(\sqrt{\frac{n}{m}}+2\right).\label{exp1}
\end{align}

The inequality in (\ref{exp1}) holds with probability $\P(\mathcal{E}_0\cap\mathcal{E}_1\cap\mathcal{E}_2)$. In the last step of the proof,  a lower bound on this probability or equivalently an upper bound on $\P(\mathcal{E}^c_0\cup\mathcal{E}^c_1\cup\mathcal{E}^c_2)$ is derived.

Fixing  $X^n=x^n$ and $\xt^n$, $A(x^n-\xt^n)/\|\x^n-\xt^n\|_2$ is a vector of i.i.d. $\mathcal{N}(0,1)$ random variables. Therefore, by Lemma \ref{chi2},
\begin{align*}
   \P_A&\left(\|A(x^n-\xt^n)\|_2\leq\sqrt{(1-\tau)m} \|x^n-\xt^n\|_2\right)\\
   &\;\;\leq\mathrm{e}^{\frac{m}{2}(\tau+\ln(1-\tau))},
\end{align*}
and  by the union bound, for a fixed $X^n=x^n$,
\begin{align}
   \P_A&(\exists \xt^n\in\Cc:\;\|A(x^n-\xt^n)\|_2\leq\sqrt{(1-\tau)m} \|x^n-\xt^n\|_2 )\nonumber\\
   &\;\;\leq 2^{nR}\mathrm{e}^{\frac{m}{2}(\tau+\ln(1-\tau))}\nonumber\\
   &\;\;=2^{nR+{m\over 2}(\tau \log\ex +\log(1-\tau))}.\label{eq:Ec1_0}
\end{align}
%Now consider taking the randomness in the generation of $X^n$ into account as well. 
Taking the expected value of the both sides of \eqref{eq:Ec1_0} with respect to $X^n$, and noting that the right hand side  of \eqref{eq:Ec1_0} is not random, it follows that
\begin{align}
 \E_{X^n}&\Big[ \P_A(\exists \xt^n\in\Cc:\nonumber\\
 &\;\;\; \|A(X^n-\xt^n)\|_2\leq\sqrt{(1-\tau)m} \|X^n-\xt^n\|_2 )\Big]\nonumber\\
 &\leq 2^{nR+{m\over 2}(\tau \log\ex +\log(1-\tau))}.\label{eq:Ec1_1}
\end{align}

Rewriting $\P_A(\Ec)$ as $\E_A[\ind_{\Ec}]$ and employing Fubini's theorem to exchange the order of integration, it follows from \eqref{eq:Ec1_1} that
\begin{align}
 \E_{A}&\Big[  \P_{X^n}(\exists \xt^n\in\Cc:\nonumber\\
 &\;\;\; \|A(X^n-\xt^n)\|_2\leq\sqrt{(1-\tau)m} \|X^n-\xt^n\|_2 )\Big]\nonumber\\
 &\leq 2^{nR+{m\over 2}(\tau \log\ex +\log(1-\tau))}.\label{eq:Ec1_2}
\end{align}

Substituting  for $m$, the exponent in \eqref{eq:Ec1_2} can be upper-bounded as follows:
\begin{align*}
nR+&{m\over 2}(\tau \log\ex +\log(1-\tau))\\
&= nR \Big( 1+\frac{\eta}{-\log D} (1-D^{(1+\d)/\eta}+\frac{(1+\d)}{\eta}\log D) \Big)\\
&= nR \Big( 1-(1+\d)- \frac{ \eta}{\log D} (1-D^{(1+\d)/\eta}) \Big)\\
&\leq nR (-\d - \frac{ \eta}{\log D})\nonumber\\
&=-nR\alpha.\label{eq:ub-exponent}
\end{align*}

Define the function $\upsilon_n$, where $\upsilon_n: \mathds{R}^{m\times n}\to[0,1]$, as $\upsilon_n(A)\triangleq \P_{X^n}(\exists \xt^n\in\Cc:\;\|A(X^n-\xt^n)\|_2\leq\sqrt{(1-\tau)m} \|X^n-\xt^n\|_2 )$. Note that in our model, $m$ is also a function of $n$. We prove that $\upsilon(n)$ converges to zero, almost surely. By Markov's inequality,  from \eqref{eq:Ec1_2} and \eqref{eq:ub-exponent},  it follows that
\begin{align}
\P(\upsilon_n(A)>2^{-{1\over 2}nR\alpha})\leq {\E[\upsilon_n(A)]\over 2^{-{1\over 2}nR\alpha}}\leq 2^{-{1\over 2}nR\alpha}.
 \end{align}
 Therefore, by the Borel Cantelli Lemma, $\upsilon_n(A)<2^{-{1\over 2}nR\alpha}$, eventually almost surely, and hence $\upsilon_n(A)$ converges to zero, almost surely. This results implies that with probability one $\P_{X^n}(\Ec_1^c)$ converges to zero.

Finally, to upper bound $\P(\mathcal{E}^c_2)$, from  \cite{CaTa05},  by the concentration of Lipschitz functions of a Gaussian vector,
\begin{align}
\P(\sigma_{\max}(A)-\sqrt{m}-\sqrt{n}\geq t\sqrt{m})   \leq\mathrm{e}^{-mt^2/2}.\label{eq:Ec2-1}
\end{align}
Letting $t=1$ in \eqref{eq:Ec2-1}, it follows that
\begin{align}
   \P(\mathcal{E}^c_2)&=\P(\sigma_{\max}(A)-\sqrt{m}-\sqrt{n}\geq \sqrt{m})\nonumber\\
   &\leq \mathrm{e}^{-m/2}.\label{eq:Ec2-final}
\end{align}

\subsection{Proof of Theorem \ref{Noisy Thm}}

Let $\hat{\X}^n=g_n(f_n(\X^n))$. Similar to the proof of Theorem \ref{NoiselessThm},  since  $\tilde{\X}^n=\argmin_{x^n\in\Cc_n}\|Y^m-Ax^n\|_2^2$, and $\Xh^n\in\Cc_n$, $\|Y^m-A\tilde{\X}^n\|_2\leq\|Y^m-A\hat{\X}^n\|_2$. Substituting $Y^m=A\X^n+Z^m$, we have
  \begin{align*}
&   \|AX^n-A\tilde{\X}^n\|_2-\|Z^m\|_2\nonumber\\
   &\leq\|A\X^n-A\hat{\X}^n\|_2+\|Z^m\|_2,
\end{align*}
or
\begin{align}
   \|AX^n-A\tilde{\X}^n\|_2&\leq\|A\X^n-A\hat{\X}^n\|_2+2\|Z^m\|_2 .\label{eq:triangle-ineq-noise}
\end{align}
Define events $\Ec_0$, $\Ec_1$ and $\Ec_2$ as in the proof of Theorem \ref{NoiselessThm} and $\Ec_3\triangleq \{{1\over \sqrt{m}}\|Z^m\|_2\leq \sigma_m\}$. Following similar steps as before, conditioned on $\Ec_0\cap\Ec_1\cap\Ec_2\cap\Ec_3$, we have
\begin{align}
   \sqrt{m(1-\tau)}\|X^n-\tilde{\X}^n\|_2
   &\leq (\sqrt{n}+2\sqrt{m})\sqrt{nD}+2\|Z^m\|_2\nonumber\\
   &\leq (\sqrt{n}+2\sqrt{m})\sqrt{nD}+2\sigma_m\sqrt{m}.\label{eq:cond-on-E1-2-3-noise}
\end{align}
The rest of the proof follows by setting $\tau=1-D^{(1+\d)/\eta}$, and substituting the values of the parameters in \eqref{eq:cond-on-E1-2-3-noise}.

\subsection{Proof of Lemma \ref{lemma:connect-UID-URDD}}
Given $k$, define distance measure $\rho_k$ such that for $x^k,\xh^k\in\mathds{R}^k$, $\rho_k(x^k,\xh^k)\triangleq \sqrt{kd_k(x^k,\xh^k)}$ where $d_k(\cdot,\cdot)$ is defined in (\ref{sq-err-distortion}). Note that $(\mathds{R}^k,\rho_k)$ is a metric space. Furthermore, since $\max_{i=1}^k|x_i-\xh_i|\leq \rho_k(x^k,\xh^k)\leq \sqrt{k} \max_{i=1}^k|x_i-\xh_i|$, from Theorem \ref{thm:prop3-3},
\[
2\limsup_{D\to0}{kR^{(k)}({\Xbbf},{D\over k})\over \log{1\over D}}=\bar{d}(X^k).
\]
By a change of variable, $2\limsup_{D\to0}{kR^{(k)}({\Xbbf},D)\over \log{1\over D}+\log{1\over k}}=\bar{d}(X^k),$
or
\[
2\limsup_{D\to0}{R^{(k)}({\Xbbf},D)\over \log{1\over D}}={1\over k}\bar{d}(X^k).
\]
Taking the limit of both sides as $k$ grows to infinity, and employing Lemma 2 from \cite{JalaliP:14-arxiv}, which shows that the upper ID of a process $\Xbbf$ can be alternatively be represented as
\[
\bar{d}_o({\Xbbf})=\lim_{k\to\infty}{1\over k}\bigg(\limsup_{b\to\infty} {H([X^k]_b)\over b}\bigg),
\]
yields
\begin{align}
\lim_{k\to\infty}\bigg(2\limsup_{D\to0}{R^{(k)}({\Xbbf},D)\over \log{1\over D}}\bigg)&=\lim_{k\to\infty}{1\over k}\bar{d}(X^k)\nonumber\\
&=\bar{d}_o({\Xbbf}).\label{eq:application-lemma-1}
\end{align}

Since $R^{(k)}({\Xbbf},D)\geq \inf_mR^{(m)}({\Xbbf},D)$, from \eqref{eq:application-lemma-1},
\begin{align*}
\bar{d}_o({\Xbbf})&\geq \lim_{k\to\infty}\bigg(2\limsup_{D\to0}{\inf_m R^{(m)}({\Xbbf},D)\over \log{1\over D}}\bigg)\nonumber\\
&\stackrel{(a)}{=} \lim_{k\to\infty}\bigg(2\limsup_{D\to0}{R({\Xbbf},D)\over \log{1\over D}}\bigg)=\overline{\dim}_{R}(\Xbbf),
\end{align*}
where (a) follows from the fact that $R(\Xbbf,D)=\inf_{m} R^{(m)}(\Xbbf,D)$ \cite{Gallager}. This proves the lower bound in the desired result.

To prove the upper bound, fix a positive integer $m\in\mathds{N}$. Any integer $k$ can be written as $k=sm+r$, where $r\in\{0,\ldots,m-1\}$. Since $kR^{(k)}({\Xbbf},D)$ is  a sub-additive sequence \cite{Gallager}, $kR^{(k)}({\Xbbf},D)\leq sm R^{(m)}({\Xbbf},D)+rR^{(r)}({\Xbbf},D),$
or
\begin{align}
R^{(k)}({\Xbbf},D)\leq {sm\over k} R^{(m)}({\Xbbf},D)+{r\over k}R^{(r)}({\Xbbf},D).\label{eq:sub-additive}
\end{align}
Combining \eqref{eq:application-lemma-1} and \eqref{eq:sub-additive}, it follows that
\begin{align}
\bar{d}_o({\Xbbf})\leq &\; 2\lim_{k\to\infty}\bigg(\limsup_{D\to0} {sm\over k} { R^{(m)}({\Xbbf},D)\over \log {1\over D}}\bigg)\nonumber\\
&+2\lim_{k\to\infty}\bigg( \limsup_{D\to0}{r\over k}{R^{(r)}({\Xbbf},D)\over \log{1\over D}}\bigg)\nonumber\\
= &\; 2\lim_{k\to\infty}\Big({sm\over k}\Big)\bigg(\limsup_{D\to0}  { R^{(m)}({\Xbbf},D)\over \log {1\over D}}\bigg)\nonumber\\
&+2\lim_{k\to\infty}\Big({r\over k}\Big)\bigg( \limsup_{D\to0}{R^{(r)}({\Xbbf},D)\over \log{1\over D}}\bigg)\nonumber\\
= &\; 2\bigg(\limsup_{D\to0}  { R^{(m)}({\Xbbf},D)\over \log {1\over D}}\bigg).\label{eq:ub_wo_inf}
\end{align}
Since $m$ is selected arbitrarily, we can take infimum of the  right hand side of \eqref{eq:ub_wo_inf} and derive the desired result.

%--------------------%--------------------%--------------------%--------------------
%--------------------%--------------------%--------------------%--------------------

\subsection{Proof or Lemma \ref{lemma:uniform_conv}}

By the lemma's assumption, $\overline{\dim}_{R}(\Xbbf) ={\dim}_{R}(\Xbbf) $; therefore, from Lemma \ref{lemma:connect-UID-URDD}, \begin{align}\label{eq:b1}
{\dim}_{R}(\Xbbf) \leq \bar{d}_o({\Xbbf})\leq 2\Big(\lim_{D\to0}  { R^{(m)}({\Xbbf},D)\over \log {1\over D}}\Big),
\end{align}
for all $m$. Given the uniform convergence assumption, for any $\e>0$, there exists $m_{\e}\in\mathds{N}$, such that for all $m>m_{\e}$,
\begin{align}\label{eq:cond1}
\left|{R^{(m)}({\Xbbf},D)\over \log{1\over D}}-{R({\Xbbf},D)\over \log{1\over D}}\right|<\epsilon,
\end{align}
for all $D\in(0,\sigma^2_{\max})$.

On the other hand, for any $\e'>0$ and  $m$, there exists $\d_{\e',m}>0$, such that for all $D\in(0,\d_{\e',m})$,
\begin{align}\label{eq:cond2}
\lim_{D\to0}  { R^{(m)}({\Xbbf},D)\over \log {1\over D}}\leq { R^{(m)}({\Xbbf},D)\over \log {1\over D}} +\e'.
\end{align}
Also, for any $\e''>0$, there exists $\d_{\e''}>0$, such that for all $D\in(0,\d_{\e''})$,
\begin{align}\label{eq:cond3}
 { R({\Xbbf},D)\over \log {1\over D}} \leq \frac{1}{2}\left({\dim}_{R}(\Xbbf) +\e''\right).
\end{align}

Therefore, for any $\e,\e'$ and $\e''$, choosing $m>m_{\e}$, and $D\in(0,\min(\d_{\e',m},\d_{\e''}))$, and combining \eqref{eq:cond1}, \eqref{eq:cond2} and \eqref{eq:cond3} yields
\begin{align}\label{eq:b2}
\bar{d}_o({\Xbbf})\leq  {\dim}_{R}(\Xbbf) +\e+\e'+\e''.
\end{align}
Since $\e,\e'$ and $\e''$ are selected arbitrarily, combining \eqref{eq:b1} and \eqref{eq:b2} proves that ${\dim}_{R}(\Xbbf) = \bar{d}_o({\Xbbf})$.

\subsection{Proof of Theorem \ref{thm:ID-eq-RDD}}

It is shown in \cite{WynerZ:71} that for any stationary process $\Xbbf$
\begin{align}
|R^{(m)}(\Xbbf,D)-R(\Xbbf,D)|\leq{1\over m}I(X^m;X^0_{-\infty}). \label{eq:bd-Wyner-Ziv}
\end{align}
Note that while some of the results in \cite{WynerZ:71} only hold for sources that are either absolutely continuous or discrete, as shown in Appendix \ref{app:WZ-genaral}, this bound holds for general sources. Since the right hand side of \eqref{eq:bd-Wyner-Ziv} does not depend on $D$, it shows that $R^{(m)}(\Xbbf,D)$ uniformly converges to $R(\Xbbf,D)$ for all $D>0$. On the other hand, for any $0<\sigma_{\max}<1$, and any $D\in(0,\sigma_{\max}^2)$, $0<1/\log{1\over D}<1/\log{1\over \sigma_{\max}^2}$. Therefore, ${R^{(m)}(\Xbbf,D)\over \log{1\over D}}$ uniformly converges  to ${R(\Xbbf,D)\over \log{1\over D}}$, for $D\in(0,\sigma_{\max}^2)$, and by Lemma \ref{lemma:uniform_conv}, ${\dim}_{R}(\Xbbf) = \bar{d}_o({\Xbbf}).$

%--------------------------%------------------------------------%-----------------------
%--------------------------%------------------------------------%-----------------------

\subsection{Proof of Theorem \ref{thm:r-d-bound}}\label{proof:thm-r-d-bound}

Let $X^n$ denote the output of the source. Given the source model, $X^n$ can be written as
\[
X^n=\underbrace{S_1,\ldots,S_{1}}_{T_1}, \underbrace{S_2,\ldots,S_{2}}_{T_2}, \ldots, \underbrace{S_N,\ldots,S_{N}}_{T_N},
\]
where $S_1,S_2,\ldots,S_N$ are i.i.d.~distributed according to $f_c$, and $\sum_{i=1}^{N}T_i=n.$ Moreover, $T_1,\ldots,T_{N-1}$ are i.i.d.~distributed geometric random variables with parameter $p$. That is, for $i=1,\ldots,N-1$ and $m\geq 1$, $ \P(T_i=m)=(1-p)^{m-1}p.$

%--------------------------%------------------------------------%-----------------------

\subsubsection{Converse}

Assume that  the pair $(R,D)$ is achievable for the coding source $X$. Then for any $\e>0$, there exists a code of blocklength $n$ sufficiently large, which operates at rate $R$ and achieves distortion $D+\e$. We prove that $R\geq pR_{f_c}(D)$:
\begin{align}
	nR&\geq H(M) \geq I(M;\Xh^n)\nonumber\\
	&\geq I(X^n;\Xh^n)=I(S^N,T^N,N;\Xh^n)\nonumber\\
	&\geq I(S^N;\Xh^n\mid T^N,N)\nonumber\\
%	&=h(S^N\mid T^N)-h(S^N\mid\Xh^n,T^N)\nonumber\\
	&=h(S^N\mid T^N,N)-h(S^N\mid\Xh^n,T^N,N) \nonumber\\
	&=\sum_{k=1}^n p_N(k)\bigg(h(S^k\mid T^k,N=k)\nonumber\\
	&\;\;\;\;\;\;\;\;\;\;\;\;\;\;\;\;\;\;\;\;\;-h(S^k|\Xh^n,T^k,N=k)\bigg)\nonumber\\
	&=\sum_{k=1}^n p_N(k)\left(h(S^k)-h(S^k\mid\Xh^n,T^k)\right)\nonumber\\
	&=\sum_{k=1}^n p_N(k)\left(\sum_{i=1}^k\left(h(S_i)-h(S_i\mid S^{i-1},\Xh^n,T^k)\right)\right)\nonumber\\
	&\geq\sum_{k=1}^n p_N(k)\left(\sum_{i=1}^k\left(h(S_i)-h(S_i\mid \Xh_{L_{i}}^{L_{(i+1)}-1},T^k)\right)\right)\label{Eq-nR:a}\\
	&=\sum_{k=1}^n p_N(k)\left(\sum_{i=1}^k\left(I(S_i; \Xh_{L_{i}}^{L_{(i+1)}-1}\mid T^k)\right)\right)\label{Eq-nR:b}
\end{align}
where in \eqref{Eq-nR:a} $L_i=1+\sum_{j=1}^{i-1}T_j$ and \eqref{Eq-nR:b} holds because $S$ and $T$ are independent. Given $T^k$ define $\Sh_i$ as follows:
\[
	\Sh_i = \argmin_{x\in\{\Xh_j  : j=L_i,\dots,L_{(i+1)-1}\}}d(S_i,x)
\]
Hence,
\begin{align}
	nR&\geq\sum_{k=1}^n p_N(k)\left(\sum_{i=1}^k\left(I(S_i; \Xh_{L_{i}}^{L_{(i+1)}-1}\mid T^k)\right)\right)\nonumber\\
	&\geq\sum_{k=1}^n p_N(k)\left(\sum_{i=1}^k\left(I(S_i; \Sh_i\mid T^k)\right)\right)\nonumber\\
	&=\sum_{k=1}^n p_N(k)\left(\sum_{i=1}^k\left(I(S_i; \Sh_i T^k)\right)\right)\label{Eq-nR2:1}\\%\label{Eq-nR2:a}\\
	&\geq\sum_{k=1}^n p_N(k)\left(\sum_{i=1}^k\left(I(S_i; \Sh_i)\right)\right)\nonumber\\
	&{\geq}\sum_{k=1}^n p_N(k) \sum_{i=1}^kR_{f_c}(\E[d(S_i,\Sh_i)])\label{Eq-nR2:2}\\
		&=\sum_{k=1}^n kp_N(k){1\over k} \sum_{i=1}^kR_{f_c}(\E[d(S_i,\Sh_i)])\nonumber\\
		&{\geq} \sum_{k=1}^n kp_N(k)R_{f_c}({1\over k}\sum_{i=1}^k \E[d(S_i,\Sh_i)])\label{Eq-nR2:3}\\
				&= \sum_{k=1}^n kp_N(k)R_{f_c}(\E[d_k(S^N,\Sh^N)|N=k]) \nonumber\\%\label{Eq-nR2:b}\\
				&= \E[NR_{f_c}(\E[d_N(S^N,\Sh^N)])],\label{Eq-nR2:c}
\end{align}
where  step \eqref{Eq-nR2:1} follows from the independence of $S_i$ and $T^k$ for all $i$, step \eqref{Eq-nR2:2} uses the definition of the rate-distortion function for source $S$, and  step \eqref{Eq-nR2:3} follows from the convexity of $R_{f_c}(D)$ and Jensen's inequality.
%\begin{itemize}
%\item step $(a)$ follows from the independence of $S_i$ and $T^k$ for all $i$.
%\item step $(b)$ uses the definition of the rate-distortion function for source $S$,
%\item step $(c)$ follows from the convexity of $R_{f_c}(D)$ and the Jensen inequality.
%\end{itemize}
On the other hand, given that $N=k$,
\begin{align}
{1\over n}d(X^n; \Xh^n)&={1\over n}\sum_{i=1}^k \sum_{j=L_{i}}^{L_{i+1}-1}d(X_j,\Xh_j)\nonumber\\
&\geq {1\over n}\sum_{i=1}^k \sum_{j=L_{i}}^{L_{i+1}-1}d(S_i,\Sh_i)\nonumber\\
&= {1\over n}\sum_{i=1}^k T_id(S_i,\Sh_i).
\end{align}

Taking expectations on both sides, it follows that
\begin{align}
\E[d_n(X^n; \Xh^n)]&\geq  \E[{1\over n}\sum_{i=1}^N T_i d(S_i,\Sh_i)]\nonumber\\
&\geq \sum_{k=1}^{n} \E[{1\over n}\sum_{i=1}^k T_i d(S_i,\Sh_i)|N=k]p_N(k)\nonumber\\
&= {1\over n}\sum_{k=1}^{n}\sum_{i=1}^k \bigg(\E[T_i|N=k] \nonumber\\
&\;\;\;\;\;\;\;\;\;\;\;\;\;\;\;\;\E[d(S_i,\Sh_i)|N=k]p_N(k)\bigg).\label{eq:ineq-E-dn}
\end{align}
 Note that $T_1,T_2,\ldots,T_{N-1}$ are i.i.d.~and there exists $\tilde{T}_N$ such that $T_1,T_2,\ldots,T_{N-1},\tilde{T}_N$ are i.i.d. and $\sum_{i=1}^{N-1}T_i+\tilde{T}_N\geq n$. Given $N=k$, $\E[T_1|N=k]=\ldots=\E[T_{k-1}|N=k]=\E[\tilde{T}_k|N=k]$, and therefore
 \begin{align}
 \E[T_i | N=k]\geq {n\over k}.\label{eq:exp-Ti-given-N-eq-k}
 \end{align}
 Combining \eqref{eq:ineq-E-dn} and \eqref{eq:exp-Ti-given-N-eq-k}, since $\E[ T_k d(S_k,\Sh_k)|N=k]\geq 0 $, it follows that
 \begin{align}
\E[d_n(X^n; \Xh^n)] &\geq  {1\over n}\sum_{k=1}^{n}\sum_{i=1}^{k-1} {n\over k} \E[d(S_i,\Sh_i)|N=k]p_N(k)\nonumber\\
&= \E\bigg[ {N-1\over N} d_{N-1}(S^{N-1},\Sh^{N-1})\bigg].\label{eq:ineq-E-dn-1}
\end{align}
 But, $Nd_{N}(S^{N},\Sh^{N})=(N-1)d_{N-1}(S^{N-1},\Sh^{N-1})+d(S_N,\Sh_N).$ Hence,
 \begin{align}
 \E[d_{N}(S^{N},\Sh^{N})]&\leq \E\bigg[{N-1\over N}d_{N-1}(S^{N-1},\Sh^{N-1})\bigg]\nonumber\\
 &\;\;+E\bigg[{d_{\max}\over N}\bigg].\label{eq:ineq-E-dn-2}
\end{align}
Combining \eqref{eq:ineq-E-dn-1} and \eqref{eq:ineq-E-dn-2} yields
 \begin{align}
\E[d_n(X^n; \Xh^n)] &\geq   \E[d_{N}(S^{N},\Sh^{N})]-E\bigg[{d_{\max}\over N}\bigg].\label{eq:ineq-E-dn-3}
\end{align}

Since $N$ counts the number of jumps in $X^n$, it can be written as $\sum_{i=1}^n\ind_{X_i\neq X_{i-1}}.$
Let $U_i=\ind_{X_i\neq X_{i-1}}$. By construction, $\{U_i\}_{i=1}^{n}$ is a sequence of i.i.d.~$\Bern(p)$ random variables. Therefore, by Hoeffding's inequality \cite{Hoeffding},
\begin{align}
\P(|{1\over n}\sum_{i=1}^n U_i-p|>\e_1)\leq 2\ex^{-2n\e_1^2}.\label{eq:P-Ec2-c-l}
\end{align}
Now let $\e_n={p\over n^{1/4}},$
and define the event $\Ec_1$ as 
\begin{align}
\Ec_1=\{|{N\over n}-p|<\e_n\}. \label{eq:P-Ec1}
\end{align}
Conditioning on $\Ec_1$ we can rewrite \eqref{eq:ineq-E-dn-3} as
\begin{align}
\E[d_n(X^n; \Xh^n)] \geq&   \E[d_{N}(S^{N},\Sh^{N})]\nonumber\\
&\;\;\;-P(\Ec_1){d_{\max}\over n(p-\e_n)}-P(\Ec_1^c){d_{\max}}\nonumber\\
\geq&\E[d_{N}(S^{N},\Sh^{N})]\nonumber\\
&\;\;\;-{d_{\max}\over n(p-\e_n)}-2\ex^{-2n\e_n^2}d_{\max}\nonumber\\
=&\E[d_{N}(S^{N},\Sh^{N})]-\d_n,\label{eq:ineq-E-dn-4}
\end{align}
where $\d_n\to 0$ as $n\to\infty$. Combining \eqref{Eq-nR2:c} and \eqref{eq:ineq-E-dn-4} yields
\begin{align}
R&\geq  \E\bigg[{N \over n}R_{f_c}(\E[d_N(S^N,\Sh^N)|N])\bigg]\nonumber\\
&=  \E\bigg[{N \over n}R_{f_c}(\E[d_N(S^N,\Sh^N)|N])|\Ec_1\bigg]\P(\Ec_1)\nonumber\\
&\;\;\;\;+ \E\bigg[{N \over n}R_{f_c}(\E[d_N(S^N,\Sh^N)|N])|\Ec^c_1\bigg]\P(\Ec_1^c)\nonumber\\
&\geq \E\bigg[{N \over n}R_{f_c}(\E[d_N(S^N,\Sh^N)|N])|\Ec_1\bigg]\P(\Ec_1)\nonumber\\
&\geq (p-\e_n) \E[R_{f_c}(\E[d_N(S^N,\Sh^N)|N])|\Ec_1]\P(\Ec_1)\nonumber\\
&= (p-\e_n) \sum_{k=n(p-\e_n)}^{n(p+\e_n)}p_N(k)R_{f_c}(\E[d_k(S^k,\Sh^k)|N=k])\nonumber\\
&\geq {(p-\e_n ) \P(\Ec_1)}R_{f_c}(\E[d_N(S^N,\Sh^N)|\Ec_1]),
\end{align}
where the last step follows from Jensen's inequality. Now we already know that $\P(\Ec_1)$ is very close to one. Also, from \eqref{eq:ineq-E-dn-4},%(1.11),
\begin{align}
\E[d_n(X^n; \Xh^n)] +\d_n  &\geq \E[d_{N}(S^{N},\Sh^{N})]\nonumber\\
&\geq \E[d_{N}(S^{N},\Sh^{N})|\Ec_1]\P(\Ec_1).
\end{align}
Therefore,
\[
\E[d_{N}(S^{N},\Sh^{N})|\Ec_1] \leq {\E[d_n(X^n; \Xh^n)] +\d_n   \over \P(\Ec_1)},
\]
which again since  $\P(\Ec_1)$ is close to one yields the desired result.

%--------------------------%------------------------------------%-----------------------

\subsubsection{Achievability}

Consider the following encoder: to encode $X^n$, first describe $T_1,\ldots,T_{N}$ losslessly and then lossy encode $S_1,\ldots,S_N$. Assuming that the decoder already knows the blocklength $n$, to convey  $T_1,\ldots,T_{N}$ to the decoder, it suffices to code $T_1,\ldots,T_{N-1}$, because $T_{N}=n-\sum_{i=1}^{N-1}T_i$. To losslessly  describe $T_1,\ldots,T_{N-1}$, the encoder first  encodes $N$ using  the Elias gamma code \cite{Elias:75}. Since $N\in\{1,\ldots,n\}$, this requires at most $2\lfloor\log n\rfloor +1$ bits. Also, as showed earlier in \eqref{eq:P-Ec2-c-l}, $\P(|{1\over n}N-p|>\e_1)\leq 2\ex^{-2n\e_1^2}$.

Define $\e_n$ and $\Ec_1$ as in \eqref{eq:P-Ec1} in the converse part. Consider a family of lossless compression codes $(n_1,\encod_{n_1}^{(\T)},\decod_{n_1}^{(\T)})$ for the i.i.d.~source $T=\{T_i\}_{i=1}^{\infty}$, operating at rate $H(T)+\e^{(T)}_{n_1}$, $\e^{(T)}_{n_1}>0$, such that $\P(T^{n_1}\neq\hat{T}^{n_1})\to 0,$  as $n_1\to \infty$, where $\hat{T}^{n_1}=\decod_{n_1}^{(\T)}(\encod_{n_1}^{(\T)}(T^{n_1}) )$ and $\lim_{n_1\to\infty}\e^{(T)}_{n_1}=0$. By Shannon's lossless compression theorem, there exists such a family of codes satisfying these conditions  \cite{cover}. Note that $H(\T)=\sum_{m=1}^{\infty}(1-p)^{m-1}p\log((1-p)^{m-1}p)
=\frac{H(p)}{p}.$
After describing  $N$ to the decoder, if $\Ec_1$ holds,  the encoder employs the $(N-1,\encod_{N-1}^{(\T)},\decod_{N-1}^{(\T)})$ code to losslessly convey $T_1,\ldots,T_{N-1}$ to the decoder.  This requires
$(N-1)(H(T)+\e^{(T)}_{N-1}),$
bits. If $\Ec_1$ does not hold, it sends nothing else. Since the decoder knows $N$, it can determine whether $\Ec_1$ holds or not.
Define the event $\Ec_2$ as $\Ec_2=\{T^N=\hat{T}^N\}.$

The last encoding step is, conditioned on $\Ec_1$ holding, to describe $S_1,\ldots,S_N$.  Let $(n_2,\encod_{n_2}^{(S)},\decod_{n_2}^{(S)})$ be a family of lossy compression codes for the i.i.d.~source $S=\{S_i\}_{\indx=0}^{\infty}$ operating at rate $R_{f_c}(D)$, and expected distortion not exceeding  $D+\e^{(S)}_{n_2}$,  such that $\e^{(S)}_{n_2}>0$ and $\lim_{n_2\to\infty}\e^{(S)}_{n_2}=0$.

Overall the number of transmitted bits is either equal to $2\lfloor \log n\rfloor+1$ if $\Ec_1$ does not hold, or $2\lfloor \log n\rfloor+1+(N-1)(H(T)+\e^{(T)}_{N-1})+NR_{f_c}(D),$ otherwise. In the latter case, the rate of the code can be upper bounded as
\begin{align}
{2\lfloor \log n\rfloor+1\over n}+(p+{p\over n^{1/4}})(R_{f_c}(D)+H(T)+\e_*^{(T)})\nonumber\\
=pR_{f_c}(D)+H(p)+\e_X,
\end{align}
  where $\e_*^{(T)}=\max_{|n_1-p|\leq \e_n}\e^{(T)}_{n_1}$. Hence, $\e_X$ can be made arbitrarily small by choosing $n$ large enough.

After receiving all  encoded bits, if only $N$ is transmitted to the decoder, it reconstructs the all-zero sequence. Otherwise, it outputs $\Xh^n=\underbrace{\Sh_1,\ldots,\Sh_{1}}_{\hat{T}_1}, \underbrace{\Sh_2,\ldots,\Sh_{2}}_{\hat{T}_2}, \ldots, \underbrace{\Sh_N,\ldots,\Sh_{N}}_{\hat{T}_N}.$
Note that by  construction $\hat{N}=N$, with probability one.

By the tower property,
\begin{align*}
\E&[d_n(X^n,\Xh^n)]=\sum_{n_2=1}^{\infty}\E[d_n(X^n,\Xh^n)|N=n_2]\P(N=n_2)\nonumber\\
&\leq\sum_{n_2=n(p-\e_n)}^{n(p+\e_n)}\E[d_n(X^n,\Xh^n)|N=n_2]\P(N=n_2)\nonumber\\
&\;\;\;\;\;\;\;\;\;+ d_{\max}\P(\Ec_1^c)\nonumber\\
&\leq  \;\sum_{n_2=n(p-\e_n)}^{n(p+\e_n)}\E[d_n(X^n,\Xh^n)|N=n_2,\Ec_2]\P(N=n_2,\Ec_2)\nonumber\\
&\;\;\;\;\;\;\;\;\; + d_{\max}\P(\Ec_1^c\cup\Ec_2^c).
\end{align*}
Conditioned on $\Ec_1\cap\Ec_2$, the distortion between the source block $X^n$, and its reconstruction $\Xh_n$ can be written as
$d_n(X^n,\Xh^n)={1\over n}\sum_{i=1}^nd(X_i,\Xh_i)
={1\over n}\sum_{k=1}^{N}T_k d(S_{k},\Sh_k).$
Therefore,
\begin{align*}
&\E[d_n(X^n,\Xh^n)] \nonumber\\
&\;\;\;\;\leq  {1\over n}\sum_{n_2=n(p-\e_n)}^{n(p+\e_n)}\bigg(\E[\sum_{k=1}^{n_2}T_k d(S_{k},\Sh_k)|N=n_2]\\
&\;\;\;\;\;\;\;\;\;\;\;\;\;\;\;\;\;\;\;\;\;\;\;\;\P(N=n_2,\Ec_2)\bigg)+ d_{\max}\P(\Ec_1^c\cup\Ec_2^c).
\end{align*}
Conditioned on $N$, $T_k$ and $d(S_{k},\Sh_k)$ are independent, and  $T_1,\ldots,T_{N-1}$~are i.i.d. Also, there exists $\tilde{T}_N$ such that $T_N\leq \tilde{T}_N$, and $T_1,\ldots,T_{N-1},\tilde{T}_N$ are all i.i.d. Therefore,
\begin{align}
&\E[d_n(X^n,\Xh^n)] \nonumber\\
&\leq  {1\over n}\sum_{n_2=n(p-\e_n)}^{n(p+\e_n)}\bigg((n_2-1)\E[d_{n_2-1}(S^{n_2},\Sh^{n_2})]\E[T_1|N=n_2]\nonumber\\
&\;\;\;\;+\E[d(S_{n_2},\Sh_{n_2})]\E[\tilde{T}_N|N=n_2]\bigg)\P(N=n_2,\Ec_2)\nonumber\\
& \;\;\;\;+ d_{\max}\P(\Ec_1^c\cup\Ec_2^c)\nonumber\\
 &\leq  {1\over n}\sum_{n_2=n(p-\e_n)}^{n(p+\e_n)}\bigg((n_2-1)(D+\e^{(S)}_{n_2-1})\E[T_1|N=n_2]\nonumber\\
 &\;\;\;\;+\E[d(S_{n_2},\Sh_{n_2})]\E[\tilde{T}_N|N=n_2]\bigg) \P(N=n_2,\Ec_2)\nonumber\\
& \;\;\;\;+ d_{\max}\P(\Ec_1^c\cup\Ec_2^c).\label{eq:E-d-Xn-Xhn}
\end{align}
On the other hand, since $T_1,\ldots,T_{N-1}$ are i.i.d., we have $\E[T_1|N=n_2]=\ldots=\E[T_{N-1}|N=n_2].$
But $\sum_{i=1}^{N-1}T_i\leq n$. Therefore, $\E[\sum_{i=1}^{N-1}T_i|N=n_2]=(n_2-1)\E[T_1|N=n_2]  \leq n$,
and
\begin{align}
\E[T_1|N=n_2] \leq {n\over n_2-1}.\label{eq:cond-exp-T1}
\end{align}
Also,
\begin{align}
\E[\tilde{T}_N|N=n_2]\P(N=n_2,\Ec_2)
&\leq \E[\tilde{T}_N|N=n_2]\P(N=n_2)\nonumber\\
&\leq \E[\tilde{T}_N]
={1\over p}.\label{eq:ub-cond-exp-TtN}
\end{align}

Hence, combining \eqref{eq:E-d-Xn-Xhn}, \eqref{eq:cond-exp-T1}  and \eqref{eq:ub-cond-exp-TtN} yields
$\E[d_n(X^n,\Xh^n)]
\leq \max_{n_2=n(p-\e_n)}^{n(p+\e_n)} %{n_2\over n_2-1}
(D+\e^{(S)}_{n_2-1})+d_{\max}({2\e_n\over p}+\P(\Ec_1^c)+\P(\Ec_2^c\cap\Ec_1))
\leq D+\d_n,$ where $\d_n\to 0$, as $n$ grows to infinity.

%--------------------------%------------------------------------%-----------------------

\subsection{Proof of Corollary \ref{CSP-RDD-Corollary}}
Since $\liminf_{n\to\infty}{m_n\over n}>2\overline{\dim}_R({\Xbbf})$, there exists $\eta>1$, such that $\liminf_{n\to\infty}{m_n\over n}>2\eta\overline{\dim}_R({\Xbbf})$. Therefore, there exists $n_{\eta}>0$, such that for all $n>n_{\eta}$, ${m_n\over n}\geq 2\eta\overline{\dim}_R({\Xbbf}).$
On the other hand, for any $\gamma>0$, there exists $D_{\gamma}>0$, such that for all $D\leq D_{\gamma}$,
\[
2{R(\Xbbf,D)\over \log{1\over D}} \leq \overline{\dim}_R({\Xbbf})+\gamma.
\]
Hence, there exists $\eta'\in(1,\eta)$, such that choosing $\gamma$ small enough, we have
\[
{m\over n}\geq {4\eta'R(\Xbbf,D)\over \log{1\over D}},
\]
for all $n>n_{\eta}$ and $D<D_{\gamma}$.

Since $\lim_{D\to 0}(2D^{\frac{1}{2}(1-\frac{1+\d}{\eta})} (\sqrt{\frac{\log{1\over D}}{4\eta R}}+2)+\sqrt{D})=0$, there exists $D_{\Delta}<D_{\gamma}$, such that
\[
2D_{\Delta}^{\frac{1}{2}(1-\frac{1+\d}{\eta})} (\sqrt{\frac{\log{1\over D_{\Delta}}}{4\eta R}}+2)+\sqrt{D_{\Delta}}<\Delta.
\] Considering a family of lossy compression codes achieving $(R(D_{\Delta}),D_{\Delta})$ and the CSP algorithm that employs this family of codes, Theorem \ref{NoiselessThm} proves the desired result.

%--------------------------%------------------------------------%-----------------------

\subsection{Proof of Theorem \ref{thm:universal-CS}}

  Let $\Xh^n=g_n(f_n(X^n)$. Since $(n,f_n,g_n)$ denotes a family of point-wise universal lossy compression codes operating at distortion level $D$, for any $\e>0$,  for all $n$ large enough,
  \[
  \P({1\over n}|f_n(X^n)| > R(\Xbbf,D)+\e)<{\e\over 2},
  \]
  and
  \[
  \P({1\over \sqrt{n}}\|X^n-\Xh^n\|_2> \sqrt{D+\e})\leq {\e\over 2}.
  \]
  Let $\Ec_1\triangleq \{{1\over n}|f_n(X^n)| \leq R(\Xbbf,D)+\e\}\cup\{ {1\over \sqrt{n}}\|X^n-\Xh^n\|_2\leq  \sqrt{D+\e} \}$. Then, $\P(\Ec_1)\leq \e$, and conditioned on $\Ec_1^c$, $f_n(X^n)$ satisfies the condition of the UCSP optimization. Therefore, conditioned on $\Ec_1^c$,
  \begin{align}\label{eq:UCSP-perform-1}
     \|Y^m-A\Xt^n\|_2&\leq \|Y^m-Ag_n(f_n(X^n))\|_2 \nonumber\\
    &\leq\sigma_{\max}(A)\|X^n-\Xh^n\|_2\nonumber\\
    &\leq \sigma_{\max}(A)\sqrt{n(D+\e)}.
  \end{align}
The rest of the proof is very similar to the proof of Theorem \ref{NoiselessThm}. The only difference is that in this case, instead of the size of the codebook, we need to bound the size of the set $\Bc=\{b: b\in\{0,1\}^*, |b|\leq n(R(\Xbbf,D)+\e) \}$. But $|\Bc|=\sum_{i=1}^{n(R(\Xbbf,D)+\e)}2^i=2^{n(R(\Xbbf,D)+\e)+1}-1$. The rest of the proof follows similar to the proof of Theorem \ref{NoiselessThm}.

%--------------------------%------------------------------------%-----------------------
\section{Conclusions}\label{sec:conc}

In this paper, we have studied the application of  rate-distortion codes in building compressed sensing  recovery algorithms for stochastic processes.  Establishing such connections between rate-distortion coding and compressed sensing potentially enables application of well-studied state-of-the-art lossy compression codes in building highly efficient compressed sensing  recovery algorithms.

We have focused on the CSP algorithm proposed  in \cite{JalaliM:13-isit} as a compression-based  compressed sensing recovery algorithm  for deterministic signals. For the CSP algorithm that employs a rate-distortion code with a certain rate $R$ and distortion $D$, we have derived an upper bound on the normalized distance between  the original vector and its reconstruction that holds with high probability.  

To analyze the asymptotic performance of the CSP algorithm when the distortion $D$ approaches zero, we have defined the RDD of stationary  processes, as a generalization of the RDD of stochastic vectors introduced in \cite{KawabataD:94}. We have proved that under some mild conditions the RDD of a stationary process is equal to its ID introduced in \cite{JalaliP:14-arxiv}. Our results have demonstrated that in the limit, as $D\to0$, for sufficiently large blocklengths $n$, CSP renders a reliable reconstruction of the source vector with almost zero-distortion, with slightly more than $n$ times the RDD of the source. This is equal to the fundamental limit of compressed sensing in memoryless stationary sources shown in \cite{WuV:10}, which proves the optimality of CSP at least in cases where the lower bounds are known.

There are two major directions that remain open for future study: the first is to design algorithms to solve the minimization problem in CSP with manageable complexity; and the second is to find the fundamental limits on compressed sensing  for general stochastic stationary sources, which would enable us to see whether CSP is always optimal, and if not how far from optimal it is.

\appendices

\section{Rate of approach of $R^{(m)}(\Xbbf,D)$ to $R(\Xbbf,D)$}\label{app:WZ-genaral}

Consider a pair of random variables $(X,\Xh)\in\Xc\times \mathcal{\Xh}$,  with alphabet sets $\Xc,\mathcal{\Xh}\subset\mathds{R}$, distributed as  $p_{X,\Xh}$, where $p_{X,\Xh}$ denotes a general measure.   For sets $\Ec\in\Xc$ and $\Fc\in\mathcal{\Xh}$, the probability of the set $\Ec\times \Fc$ under $p_{X,\Xh}$ is computed as
\[
\P_{X,\Xh}(\Ec\times\Fc)=\int_{u\in\Ec\times\Fc}p_{X,\Xh}(du).
\]
The marginal distributions under $X$ and $\Xh$ are defined as
\[
\P_{X}(\Ec)=\int_{u\in\Ec\times\mathcal{\Xh}}p_{X,\Xh}(du),
\]
and
\[
\P_{\Xh}(\Fc)=\int_{u\in\Xc\times\Fc}p_{X,\Xh}(du),
\]
respectively.
Let $\Pc$ denote a partition of $\Xc\times\Fc$ into finitely many rectangles, $\{\Ec_i,\Fc_j\}_{i,j}$.
Dobrushin \cite{pinsker1964information,KawabataD:94} established that for  random variables $(X,\Xh)$ with a general distribution,  the mutual information can be generalized as
\[
I(X ;\Xh)=\sup_{\Pc}\sum_{i,j}\P_{X,\Xh}(\Ec_i\times\Fc_j)\log {\P_{X,\Xh}(\Ec_i\times\Fc_j)\over \P_{X}(\Ec_i)\P_{\Xh}(\Fc_j)}.
\]

 Wyner and Ziv in \cite{WynerZ:71} proved that with sources with either discrete or absolutely continuous distributions we  have
 \[
 \sum_{k=1}^N I(X_k;\Xh_k)-N\Delta_N-I(X^N;\Xh^N)\leq 0,
 \]
 where
 \begin{align}
 \Delta_N={1\over N}\sup_{\Pc} \sum_{i_1,\ldots,i_N}\P_{X^N}(\prod_{k=1}^N\Ec_{i_k})\log{\P_{X^N}(\prod_{k=1}^N\Ec_{i_k}) \over \prod_{k=1}^N \P_{X_k}(\Ec_{i_k})}. \label{eq:Delta_N}
 \end{align}
 In the following, we prove that  this inequality also holds for sources with general distributions. Given $(X_k,\Xh_k)$ and $\e>0$, let $\Pc_k=\{\Ec_{i_k}\times\Fc_{j_k}\}_{i_k,j_k}$ denote the partitioning of $\Xc\times\mathcal{\Xh}$ that  ensures
\begin{align*}
&I(X_k;\Xh_k)\\
&\;\;\;-\sum_{i_k,j_k}\P_{X_k,\Xh_k}(\Ec_{i_k}\times\Fc_{j_k})\log {\P_{X_k,\Xh_k}(\Ec_{i_k}\times\Fc_{j_k})\over \P_{X_k}(\Ec_i)\P_{\Xh_k}(\Fc_{j_k})}\leq {\e\over N}.
\end{align*}
Since $I(X_k;\Xh_k)$ is defined as the supremum of the objective function over all partitions, such a partition always exists. Combining these partitions yields a natural partitioning of $\Xc^N\times\mathcal{\Xh}^N$. Since to evaluate  $I(X^N;\Xh^N)$ and $\Delta_N$ involves taking suprema of the corresponding objective functions, we have
\begin{align}
&\sum_{k=1}^NI(X_k;\Xh_k)-N\Delta_N-I(X^N;\Xh^N)\nonumber\\
&\leq\sum_{k=1}^N\Big(\sum_{i_k,j_k}\P_{X_k,\Xh_k}(\Ec_{i_k}\times\Fc_{j_k})\nonumber\\
&\hspace{1.7cm}\log{\P_{X_k,\Xh_k}(\Ec_{i_k}\times\Fc_{j_k}) \over \P_{X_k}(\Ec_{i_k})\P_{\Xh_k}(\Fc_{j_k})}+{\e\over N}\Big)\nonumber\\
&\;\;-\sum_{i_1,\ldots,i_N}\P_{X^N}(\prod_{k=1}^N\Ec_{i_k})\log{\P_{X^N}(\prod_{k=1}^N\Ec_{i_k}) \over \prod_{k=1}^N \P_{X_k}(\Ec_{i_k})}\nonumber\\
&\;\;-\sum_{\substack{i_1,\ldots,i_k\\ j_1,\ldots,j_k}}\P_{X^N,\Xh^N}(\prod_{k=1}^N\Ec_{i_k}\times \prod_{k=1}^N\Fc_{j_k})\nonumber\\
&\hspace{1.7cm}\log{\P_{X^N,\Xh^N}(\prod_{k=1}^N\Ec_{i_k}\times \prod_{k=1}^N\Fc_{j_k})\over \P_{X^N}(\prod_{k=1}^N\Ec_{i_k}) \P_{\Xh^N}(\prod_{k=1}^N\Fc_{j_k})}\nonumber\\
%\end{align}
%\begin{align}
&=\e+\sum_{i^N,j^N}\P_{X^N,\Xh^N}(\prod_{k=1}^N\Ec_{i_k}\times \prod_{k=1}^N\Fc_{j_k})\nonumber\\
&\hspace{1.7cm}\log\Big[\prod_{k=1}^N{\P_{X_k,\Xh_k}(\Ec_{i_k}\times \Fc_{j_k})\over \P_{X_k}(\Ec_{i_k})\times \P_{\Xh_k}( \Fc_{j_k})}\nonumber\\
&\hspace{2.2cm}\times {\prod_{k=1}^N\P_{X_k}(\Ec_{i_k})\over \P_{X^N}(\prod_{k=1}^N\Ec_{i_k})}\nonumber\\
&\hspace{2.2cm}\times {{\P_{X^N}(\prod_{k=1}^N\Ec_{i_k})\P_{\Xh^N}}(\prod_{k=1}^N\Fc_{j_k})\over\P_{X^N,\Xh^N}(\prod_{k=1}^N\Ec_{i_k}\times\prod_{k=1}^N\Fc_{j_k})}\Big].
\end{align}
Canceling the common terms, and rearranging the terms, it follows that
\begin{align}
&\sum_{k=1}^NI(X_k;\Xh_k)-N\Delta_N-I(X^N;\Xh^N)\nonumber\\
&\leq\e+\sum_{i^N,j^N}\P_{X^N,\Xh^N}(\prod_{k=1}^N\Ec_{i_k}\times \prod_{k=1}^N\Fc_{j_k})\nonumber\\
&\;\;\log\Big({ \prod_{k=1}^N \P_{X_k|\Xh_k}(\Ec_{i_k}| \Fc_{j_k})\over \P_{X^N|\Xh^N}(\prod_{k=1}^N \Ec_{i_k} |\prod_{k=1}^N\Fc_{j_k})}\Big).\label{eq:app-1}
\end{align}
Since $\log x\leq x-1$, the right hand side of \eqref{eq:app-1} can further be upper-bounded as
\begin{align}
&\sum_{k=1}^NI(X_k;\Xh_k)-N\Delta_N-I(X^N;\Xh^N)\nonumber\\
&\leq \e+\sum_{i^N,j^N}\P_{X^N,\Xh^N}(\prod_{k=1}^N\Ec_{i_k}\times \prod_{k=1}^N\Fc_{j_k})\nonumber\\
&\hspace{1.5cm}\Big({ \prod_{k=1}^N \P_{X_k|\Xh_k}(\Ec_{i_k}| \Fc_{j_k})\over \P_{X^N|\Xh^N}(\prod_{k=1}^N \Ec_{i_k} |\prod_{k=1}^N\Fc_{j_k})}-1\Big)\nonumber\\
&= \e+\sum_{i^N,j^N}{\P_{\Xh^N}(\prod_{k=1}^N\Fc_{j_k}) \prod_{k=1}^N \P_{X_k|\Xh_k}(\Ec_{i_k}| \Fc_{j_k})}\nonumber\\
&\hspace{1.5cm}-\sum_{i^N,j^N}\P_{X^N,\Xh^N}(\prod_{k=1}^N\Ec_{i_k}\times \prod_{k=1}^N\Fc_{j_k})\nonumber\\
&=\e+1-1=\e.
\end{align}
Since $\e>0$ was selected arbitrarily, this proves the desired inequality, \ie $\sum_{k=1}^NI(X_k;\Xh_k)-N\Delta_N-I(X^N;\Xh^N)\leq 0.$ This result is analogous to Lemma 2 in  \cite{WynerZ:71}, but holds for sources with general distributions. After this generalization, the next steps required for proving the lower bound established in Section III.B of \cite{WynerZ:71} also hold in this case, with no change. Therefore,
\begin{align*}
R^{(N)}(\Xbbf,D)\geq R^{(1)}(\Xbbf,D)-\Delta_N.
\end{align*}

Using the fact that memory decreases the rate of a source \cite{WynerZ:71} we get an upper bound on $R^{(N)}(\Xbbf,D)$:
\begin{align}
R^{(1)}(\Xbbf,D)-\Delta_N\leq R^{(N)}(\Xbbf,D)\leq R^{(1)}(\Xbbf,D). \label{eq:Wyner-Ziv}
\end{align}

%%%%%%%%%%%%%%%%%%%%%%%%%%
To prove the inequality \eqref{eq:bd-Wyner-Ziv}, we first need to review some properties of $\Delta_N$. 
Following the definition in \eqref{eq:Delta_N}, it can be shown that $\Delta_N$ can be represented in terms of mutual information as follows  \cite{WynerZ:71}: 
 \begin{align}
 \Delta_N={1 \over N}\sum_{i=2}^N I(X_k;X_1^{k-1}). \label{eq:Delta_N_a}
 \end{align}

Note that with this alternative  representation it is very easy to see that $\Delta_N$ is increasing in $N$ \cite{WynerZ:71}. Putting this together with \eqref{eq:Wyner-Ziv} we get
\begin{align}
|R^{(N)}(\Xbbf,D) - R(\Xbbf,D)|\leq\Delta_N\leq\Delta_\infty, \label{eq:Wyner-Ziv-2}
\end{align}
where
\begin{align}
\Delta_\infty=\lim_{N\to\infty}\Delta_N=I(X_1;X_{-\infty}^{0}), \label{eq:Delta_infty}
\end{align}
 follows directly from \eqref{eq:Delta_N_a}. Note that $R(\Xbbf,D)$ is the rate-distortion function of the stationary process $\Xbbf$.

Let $\mathbf{Y}$ be the supersource whose outputs are successive blocks of $m$ outputs of the source $\Xbbf$. Applying \eqref{eq:Wyner-Ziv-2} to $\mathbf{Y}$ with $N=1$ we have
\begin{align*}
|R^{(1)}(\mathbf{Y},D) - R(\mathbf{Y},D)|\leq\Delta_\infty. %\label{eq:Wyner-Ziv-3}
\end{align*}
Since $\mathbf{Y}$ is defined as a supersource of successive blocks of  length $m$ of the source $\Xbbf$, it is easy to see that $R^{(1)}(\mathbf{Y},D)=mR^{(m)}(\Xbbf,D)$ and $R(\mathbf{Y},D)=mR(\Xbbf,D)$, and therefore,
\begin{align*}
|R^{(m)}(\Xbbf,D) - R(\Xbbf,D)|&\leq{1 \over m}\Delta_\infty\\
&={1 \over m}I(X_1;X_{-\infty}^{0}), 
\end{align*}
where the last line follows from \eqref{eq:Delta_infty}. Hence, the proof is complete and \eqref{eq:bd-Wyner-Ziv} holds for general stationary sources.

\section*{Acknowledgments}
This research was supported in part by the U.S. National Science Foundation grant CCF-1420575.

\bibliographystyle{unsrt}
\bibliography{../myrefs}

%%--------------------------------------Appendices--------------------------------------
\setcounter{equation}{0}
\renewcommand{\theequation}{\thesection.\arabic{equation}}

\end{document}